\documentclass[preprint,12pt]{elsarticle}

\usepackage{amssymb}
\usepackage{amsmath}
 \usepackage{amsthm}

\usepackage{cancel}
\usepackage{graphicx}

\usepackage{xcolor}
\usepackage{hyperref}
\hypersetup{
    colorlinks,
    linkcolor={red!50!black},
    citecolor={blue!50!black},
    urlcolor={blue!80!black}
}

\newtheorem{lemma}{Lemma}[section]
\newtheorem{proposition}{Proposition}[section]

\journal{Journal of Environmental Economics and Management}

\begin{document}

\begin{frontmatter}

\title{The behavioral effects of index insurance in fisheries} 

\author[ucsb]{Nathaniel Grimes}
\author[ucsb,emlab]{Christopher Costello}
\author[ucsb,emlab]{Andrew J. Plantinga}

\affiliation[ucsb]{organization={Bren School of Environmental Science and Management, University of California},
            addressline={}, 
            city={Santa Barbara},
            postcode={93106}, 
            state={CA},
            country={}}
\affiliation[emlab]{organization={Environmental Markets Lab},
            addressline={}, 
            city={Santa Barbara},
            postcode={93106}, 
            state={CA},
            country={}}            

\begin{abstract}
Fisheries are vulnerable to environmental shocks that impact stock
health and fisher income. Index insurance is a promising financial tool
to protect fishers from environmental risk. However, insurance may
change fisher's behavior. It is imperative to understand the direction
fishers change their behavior before implementing new policies as
fisheries are vulnerable to overfishing. We provide the first
theoretical application of index insurance on fisher's behavior change
to predict if index insurance will incentivize higher or lower harvests
in unregulated settings. We find that using traditional fishery models
with production variability only originating through stock abundance
leads fishers to increase harvest with index insurance. However, fishers
are adaptable and experience multiple sources of risk. Using a more
flexible specification of production shows that index insurance could
raise or lower harvest depending on the risk mitigation strategies
available for fishers and the design of the insurance contract. We
demonstrate the magnitude of potential change by simulating from
parameters estimated for three Norwegian fisheries. Fisheries with index
insurance contracts protecting extraction risks may increase harvest by
10\% or decrease by 2\% depending on the risk effects of inputs.
Insurance contracts protecting stock risk will lead to 6-20\% increases
in harvest. Before widespread adoption, careful consideration must be
given to how index insurance will incentivize or disincentivize
overfishing.
\end{abstract}


\end{frontmatter}

\section{Introduction}\label{introduction}

Fisheries are exposed to numerous environmental risks that impact
biological and economic sustainability \citep{Alvarez2006,Lehodey2006,rogers2019,Cheung2021,Oken2021,Smith2023}. There are
limited financial tools available to protect fishing communities against
environmental risks\citep{Sethi2010,Kasperski2013}. Index
insurance is a new financial tool that has significant potential to bolster community welfare in response to disastrous weather events \citep{Maltby2023,Watson2023,Hobday2025}. However, insurance possesses moral hazards that may induce behavior change in fishers. In this paper, we provide the first examination of whether index insurance will lead to incentives that increase or decrease fishery harvest. We apply our theoretical findings
to an empirical setting in three Norwegian fisheries.

Fishers are highly sensitive to risk, especially income risk, and demonstrate risk aversion despite working a seemingly risky profession \citep{Smith2005,Holland2008,Sethi2010}. Individual choices by fishers and fishery management mitigate environmental risk. Individual efforts include choosing consistent, known fishing grounds over risking
exploring unknown spots \citep{Holland2008} or choosing to fish less after storms and hurricanes \cite{Pfeiffer2020,Pfeiffer2022}. However, these efforts are unlikely to completely eliminate risk. Additional financial tools may be needed to address income risk as a
result of environmental fluctuations \citep{Sethi2010,Kasperski2013,Wabnitz2019,Sumaila2021}.

Index insurance may be an ideal financial tool for risk management in fisheries as it is scalable, protects against environmental shocks, and smooths income for fishers\citep{Watson2023,Hobday2025}. Index insurance uses independent measures of weather as the
basis for issuing payouts to all policyholders. For example, a pilot program through the Caribbean Oceans and Aquaculture Sustainability Facility (COAST) pays out a set amount to fishers when indices of wave
height, wind speed, and storm surge indicate a hurricane \citep{Sainsbury2019}.

One crucial area that remains under studied is the potential influence of insurance on fisher behavior. Moral hazards are decisions by insured
agents that they would not otherwise take if they were uninsured \cite{Wu2020}. There are two components to insurance moral hazards: ``chasing the trigger'' and ``risk reduction''. ``Chasing the trigger'' is the directed behavior of policyholders to increase the
likelihood of a payout. For example, a fisher might choose to fish less to receive an indemnified harvest insurance payment. Index insurance
completely eliminates this moral hazard as the index is independent of fisher choices, e.g.~fishers cannot affect sea surface temperature. ``Risk reduction'' occurs when policyholders possess an insurance
contract that protects them from risk, leading them to reoptimize their decisions. Index insurance remains susceptible to this element of moral
hazard. One possible response by fishers is to take on additional risk by fishing more. Alternatively, the risk protection offered by insurance
could encourage fishers to fish less as insurance payouts sufficiently cover income loss. All preliminary analyses of fisheries index insurance are missing rigorous assessment of risk reduction moral hazards.

In this paper, we assess how index insurance moral hazards could change fisher input choices. Subsequently, the changes in inputs lead to
changes in harvest, and thus sustainability. Fisheries remain vulnerable to overfishing \citep{fao2020}. It is imperative to ensure new policies, such as index insurance, do not provide perverse incentives that degrade long term sustainability by encouraging greater fishing pressures.

The current working assumption by practitioners is that index insurance will motivate fishers to fish less by relieving income constraints, but there is no evidence beyond anecdotes \citep{RARE2021}. Previous studies
articulated hypothetical examples of moral hazards in potential fishery indemnity insurance programs, such as encouraging fishers to fish in
foul weather or to not exit the fishery after a bad year of harvest \citep{Herrmann2004,Watson2023}. However,
neither study built testable models to uncover risk reduction moral hazard impacts on fisheries.

Research from agriculture provides compelling evidence that behavior change ought to be expected in fisheries. Index insurance applied to grazing in pasture commons shows clear evidence of risk reduction moral
hazards leading to environmental degradation \citep{muller2011,Bulte2021}. Other studies from agriculture find that the impact of insurance on environmental sustainability depends on the
underlying risk reducing or increasing qualities of inputs used in production \citep{Ramaswami1993,Mahul2001,Mishra2005}. Risk
increasing (decreasing) inputs will always lead to increased (decreased) input use with insurance. Numerous agricultural studies confirm insurance incentivizes changes in input use \citep{horowitz1993,Babcock1996,Smith1996,Mishra2005,Cai2016,Deryugina2017,Claassen2017,Elabed2016,Sibiko2020,stoeffler2022a,Sloggy2025}.

Fisheries differ from agriculture in crucial ways, thus motivating an analysis of the behavioral effects of index insurance in this new setting. In a standard fisheries model, production by a fisher depends
on the abundance of the resource, as measured by the fish stock. In this case, there is a positive relationship between output and inputs of
labor and capital. For example, more fishing boats results in a greater catch for a given fish stock. When the fish stock is subject to stochastic shocks, inputs will be risk-increasing and, thus, index
insurance covering biological stock risk will always increase input use and harvest of the resource. This suggests that index insurance will be in conflict with conservation goals in fisheries. However, fishers also
face risk that is independent of biological risk. For example, weather and wave conditions may make it difficult to catch fish or new regulations may limit where fishing can take place. Inputs that interact
with extraction shocks may be risk-increasing or risk-decreasing. Therefore, index insurance applied to extraction risk has the potential to reduce harvest and conserve fish stocks. A complete analysis of index insurance in fisheries is needed to understand the different types of risk faced by fishers, alternative ways in which insurance contracts can be designed, and whether insurance will be compatible with conservation
goals.

To estimate the magnitude of potential harvest changes, we numerically simulate a model from parameter estimations of Norwegian fisheries from
\citet{Asche2020}. The Norwegian fisheries show index insurance could raise harvest by 20\% or lower harvest by 2\% contingent on the fishery input characteristics and contract type. All contracts specified to protect biological stock risk showed large increases in harvest regardless of the extraction risk effects. Currently, most proposed insurance contracts are examining triggers based on stock risk, such as sea surface temperature or chlorophyll-a \citep{Watson2023}.
Without additional constraints, these types of contracts will incentivize greater exploitation of vulnerable fish stocks.

The remainder of the paper is structured as follows.
Section~\ref{sec-jp} demonstrates how assuming standard fishery production models will always bias index insurance towards overfishing.
We then present a new stochastic production function for fisheries that integrates both stock and extraction risk in Section~\ref{sec-common}. Fishers now face multiple risks where contracts could be specified to
protect either risk. In this new setting, we prove how index insurance will change fisher behavior, and how the outcome depends on the risk effects of inputs as well as the type of insurance contract.
Section~\ref{sec-multi} extends the theoretical model to account for
multiple inputs in fishing that reflects the decisions of fishers in the
empirical setting. Section~\ref{sec-sim} numerically estimates potential
harvest changes with an index insurance program. Parameters are
calibrated with an application to Norwegian fisheries through the
results of Asche \emph{et al.} (2020). Section~\ref{sec-disc} concludes
with a discussion on the suitability of fishery index insurance.

\section{Index Insurance with standard fishery production}\label{sec-jp}

Here we develop a model of fishery production under risk. We are
primarily interested in within-season input decisions and risk, so we
omit time subscripts, but will be explicit about when random variables
are known, or unknown, to fishers. We begin with a canonical model of
fishery production.

Fishers choose an input $x$ that interact through an increasing
concave production technology $f(x)$\footnote{Common specifications of
  $f(x)$ are Cobb-Douglas or linear harvest from Gordon-Schaefer.}.
Fish biomass is $B$, and total production, or ``harvest'', $y$ is
given by:

\begin{equation}\label{eq-trad}{
y = Bf(x)
}\end{equation}

We assume that $B$ is unknown prior to the choice of $x$, but that
some information on $B$ is available. Specifically, we assume
$B=\hat{B}+\theta$, where $\hat{B}$ is known and $\theta$ is mean
zero additive error term, with known variance $\sigma^2_\theta$; we
will refer to $\theta$ as the ``stock risk''. This implies that
Equation~\ref{eq-trad} can be rewritten:

\begin{equation}\label{eq-trad2}{
    y = \hat{B}f(x) + \theta f(x)
}\end{equation}

This approach combines a standard model of fishery production
(Equation~\ref{eq-trad}) with the simple observation that biomass, which
affects harvest, is unknown prior to the choice of inputs. With this
model, we can examine how index insurance affects incentives around
input choice. Under standard assumptions, fishers are price takers,
incur convex costs, and are risk averse over profits.

We can define a simple profit function where prices are set to one in
Equation~\ref{eq-pi}

\begin{equation}\label{eq-pi}{
\pi=\hat{B}f(x)+\theta f(x)-c(x)
}\end{equation}

We create an insurance lottery through a contract that uses $\theta$
as the index. The trigger $\bar{\theta}$ initiates a constant payout
$\gamma$ when the index falls below the trigger,
$\theta<\bar{\theta}$.

Actuarially fair insurance implies the premium, $\rho$, paid in both
lottery outcomes to be the probability of receiving a payout times the
payout amount, $\rho=J(\bar \theta)\gamma$, where $J(\theta)$ is the
cumulative distribution of the shock. Additionally, if we set the
trigger to zero, $\bar\theta=0$, profit will enter corresponding bad
and good states when $\theta$ is positive and negative respectively.

Fishers are expected utility maximizers. To do so, they need to make
decisions on the marginal profitability in the good and bad states. We
introduce the following lemma to help us compare marginal profits in
both states:

\begin{lemma}\label{lem-theta}

Index insurance contracts built on $\theta$ will always lead to higher
expected marginal profits in the good state

\[\frac{\mathbb{E}[\partial \pi|\theta<\bar \theta]}{\partial x}-\frac{\mathbb{E}[\partial \pi|\theta>\bar \theta]}{\partial x}<0\]

\end{lemma}

The proof of Lemma~\ref{lem-theta} is included in the appendix.

Risk aversion is a necessary condition for insurance to be desirable
\citep{Outreville2014}. Therefore, we assume fishers are risk averse to
income shocks through a concave utility function. Fishers will maximize
their own expected utility in the lottery by choosing the input given an
exogenous insurance contract (Equation~\ref{eq-max}).

\begin{equation}\label{eq-max}{
\begin{aligned}
U\equiv\max_{x}\mathbb{E}[U]=& \int^{\bar \theta}_{-\infty}j(\theta)u(\pi(x,\hat{B},\theta)+(1-J(\bar \theta))\gamma)d\theta\\
&+\int^{\infty}_{\bar{\theta}}j(\theta) u(\pi(x,\hat{B},
\theta)-J(\bar \theta)\gamma)d \theta
\end{aligned}
}\end{equation}

The first order condition that solves Equation~\ref{eq-max} is then:

\begin{equation}\label{eq-foc1}{
\begin{aligned}
\frac{\partial U}{\partial x}=&\int^{\bar \theta}_{-\infty}j(\theta)u_{x}(\pi(x,\hat{B},\theta)+(1-J(\bar \theta))\gamma)\frac{\partial \pi}{\partial x}(x,\hat{B},\theta)d\theta\\
&+\int^{\infty}_{\bar{\theta}}j(\theta) u_{x}(\pi(x,\hat{B},\theta)-J(\bar \theta)\gamma)\frac{\partial \pi}{\partial x}(x,\hat{B},\theta) d\theta\\
&=0
\end{aligned}
}\end{equation}

To find the effect of insurance on optimal input, we use the implicit
function theorem to examine how input choice varies with the insurance
payout $\gamma$:

\begin{equation}\label{eq-ivt}{
\frac{\partial x^{*}}{\partial \gamma}=-\frac{\frac{\partial U}{\partial x \partial \gamma}}{\frac{\partial^2 U}{\partial x^{2}}}
}\end{equation}

We use $\gamma$ to test insurance effects, because a marginal change
in the payout increases the value of insurance\footnote{As shown in
  Section~\ref{sec-sim}, marginal utility reaches a maximum value before
  decreasing at high $\gamma$. Therefore we assume the change in
  $\gamma$ occurs in regions of increasing marginal utility in
  $\gamma$ for simplicity, but our proof will hold for any $\gamma$.}.
Receiving more compensation in the bad states provides greater boosts to
utility for risk averse fishers.

By the sufficient condition of a maximization problem,
$\frac{\partial^2 U}{\partial x^{2}}$ is negative so we can focus
solely on the numerator to sign the effect. The numerator of
Equation~\ref{eq-ivt} is given by:

\begin{equation}\label{eq-xgam}{
\begin{aligned}
\frac{\partial U}{\partial x \partial \gamma}=\int^{\bar \theta}_{-\infty}j(\theta)u''(\pi(x,\hat{B},\theta)+(1-J(\bar \theta))\gamma)\frac{\partial \pi}{\partial x}(x,\hat{B},\theta)(1-J(\bar \theta))d\theta\\
+\int^{\infty}_{\bar{\theta}}j(\theta) u''(\pi(x,\hat{B},\theta)-J(\bar \theta)\gamma)\frac{\partial \pi}{\partial x}(x,\hat{B},\theta)(-J(\bar \theta)) d\theta
\end{aligned}
}\end{equation}

\begin{proposition}\label{prp-std}

For insurance contracts specified at trigger $\bar\theta=0$ with
production function $y=\hat{B}f(x)+\theta f(x)$, optimal harvest will
always increase with an increase in $\gamma$.

\end{proposition}

\begin{proof}
Suppose insurance fully covers the loss between states, then utility in
the good state and bad state are equal to each other so that we can
factor out like terms in Equation~\ref{eq-xgam}. For brevity, all like
terms including $\gamma$ are indicated by $u(\cdot)$.

\begin{equation}\label{eq-ps}{
\begin{aligned}
\frac{U}{\partial x \partial \gamma}&=J(\bar\theta)(1-J(\bar\theta))u''(\cdot)\\
&\times\left[\int^{\bar\theta}_{-\infty}j(\theta)\frac{\partial \pi}{\partial x}(x,\hat{B},\theta)d\theta-\int^{\infty}_{\bar{\theta}}j(\theta)\frac{\partial \pi}{\partial x}(x,\hat{B},\theta)d\theta\right ]
\end{aligned}
}\end{equation}

The first term of Equation~\ref{eq-ps} is negative by the concavity of
utility, $u''(\cdot)<0$. The second term of Equation~\ref{eq-ps} is
negative by Lemma~\ref{lem-theta}.

Therefore, an index insurance contract built assuming a standard fishery
production model will always lead to increased input use, and thus more
harvest.
\end{proof}

Fishers will choose an insurance contract that maximizes their expected
utility without market constraints. There could be market distortions
such as subsidies or limits on coverage that may prevent fishers from
choosing the optimal level. Regardless of the level of insurance
payouts, fishers will always increase input use by
Proposition~\ref{prp-std}.

Proposition~\ref{prp-std} immediately challenges the sustainability of
index insurance in fisheries. Fishers will always increase input choices
if we assume a standard fishery production model. The intuition why
fishers would increase harvest is identical to observations from
agriculture that studied similar specifications of production
\footnote{Equation~\ref{eq-trad2} is a special case of the production
  function used in \citet{Mahul2001} where $f(x)$ is exactly equal to
  $h(x)$ used in his specification.} \cite{Ramaswami1993,Mahul2001}.
Risk increasing inputs increase the variance of harvest (Just and Pope
1978). Effort is inherently ``risk-increasing'' in
Equation~\ref{eq-trad2} and leads to more variance\footnote{Note that:
  $V(y)=\sigma^2f(x)^2$, the derivative is always positive because
  $\frac{f(x)}{\partial x}>0$:
  $\frac{\partial V(y)}{\partial x}=2\sigma^2\frac{\partial f(x)}{\partial x}$}.
Insurance lowers exposure to income variance. With less marginal
exposure to risk, fishers can increase their risky input use.

However, fishers are exposed to more margins of risk and possess means
to mitigate those risks beyond ways captured in the simple canonical
model. In the next section we expand fishery production to include
multiple sources of risk, and incorporate risk mitigation strategies
present in fisheries.

\section{Index insurance with risky fishery
production}\label{sec-common}

Anecdotally and empirically, fishers make informed decisions to mitigate
risk beyond just exposure to stock risk \citep{Kirkley1998,Eggert2004,Kompas2004,Smith2005,Sethi2010,Holland2008,Pfeiffer2020,Pfeiffer2022}.
Stock risk, $\theta$, remains an inevitable source of risk that we
must maintain in any fishery model, but we expand the standard fishery
model in Equation~\ref{eq-trad2} to Equation~\ref{eq-jpfish} by adding a
second source of risk called extraction risk, $\omega$, and new ways
inputs could interact with extraction risk through $h(x)$.

\begin{equation}\label{eq-jpfish}{
y=f(x)\hat{B}+\theta f(x)+\omega h(x)
}\end{equation}

All other forms of risk not captured by stock risk are extraction risk,
$\omega$, where $\omega$ is a random variable unknown to fishers
before the choose inputs with $\mathbb{E}[\omega]=0$ and variance
$\sigma_\omega^2$. Foul weather, regulatory changes, or inherent
variability in extraction all impact fisher production. Fisher inputs
may interact with these risks through the extraction risk effect
function $h(x)$. Extraction risk effects may reduce or increase risk.
Inputs that increase variance are called risk increasing, $h_x(x)>0$,
and inputs that decrease variance are called risk decreasing,
$h_x(x)<0$ \citep{Just1978}.

We will refer to production in the form of Equation~\ref{eq-jpfish} as
``risky production''. Risky production is more flexible than the
standard fishery production model in Equation~\ref{eq-trad2} as it
allows for two sources of risk and multiple avenues for fishers to
mitigate risk. Risky production nests the standard fishery production
model as a special case when $\omega=0$ or $h(x)=0$. It also
maintains increasing mean production in inputs so that changes in inputs
correspond to changes in expected harvest and conservation.

With two sources of risk and adaptive margins, we have created a more
nuanced representation of risky production in fisheries. However, it is
not immediately clear how index insurance will affect fisher harvest
decisions with Equation~\ref{eq-jpfish}. We can amend the insurance
framework from Section~\ref{sec-jp} to assess the potential behavior
implications of an insurance contract to protect against stock,
$\theta$, or extraction risk, $\omega$.

Profits with risky production update Equation~\ref{eq-pi} to
Equation~\ref{eq-pi1}:

\begin{equation}\label{eq-pi1}{
\begin{aligned}
\pi=f(x)\hat{B}+\theta f(x)+\omega h(x)-c(x)
\end{aligned}
}\end{equation}

We assume fishers and insurance companies have perfect information on
both distributions of random variables.

We create insurance lotteries through contracts that use either
$\omega$ or $\theta$ as the trigger. For notational ease, we present
the model with a contract built on $\omega$, but the structure is
interchangeable with contracts built on $\theta$. By allowing
contracts on only one of the random variables, we introduce basis risk
as a contract triggered solely on $\omega$ cannot protect against all
the biological risk of $\theta$. We assume that $\theta$ and
$\omega$ are independent\footnote{In the appendix we include proofs
  when shocks are perfectly correlated, but the results become difficult
  to extract clear insights}. An example of how shocks could be
independent is that stock shocks operate at different time scales than
extraction shocks\citep{Alvarez2006}.

Payouts remain a constant $\gamma$ and fishers pay a premium,
$\rho=J(\bar(\omega)\gamma)$, in both states.

As before, the marginal profits in both states are crucial to
understanding fisher behavior. For contracts built on $\omega$ the
risk effects function $h(x)$ becomes influential. We introduce the
following lemma to help us compare marginal profits in both states:

\begin{lemma}\label{lem-mp}

Expected marginal profit is higher in bad states for risk decreasing
inputs when insurance contracts are built on extraction risk $\omega$.

$\frac{\mathbb{E}[\partial \pi|\omega<\bar \omega]}{\partial x}-\frac{\mathbb{E}[\partial \pi|\omega>\bar \omega]}{\partial x}>0$
if $h_{x}(x)<0$.

Otherwise, risk increasing inputs lead to higher expected marginal
profit in the good states.

$\frac{\mathbb{E}[\partial \pi|\omega<\bar \omega]}{\partial x}-\frac{\mathbb{E}[\partial \pi|\omega>\bar \omega]}{\partial x}<0$
if $h_{x}(x)>0$

\end{lemma}

The proof of Lemma~\ref{lem-mp} is included in the appendix. The results
of Lemma~\ref{lem-theta} hold even with risky production in
Equation~\ref{eq-jpfish} as shown in the proof of Lemma~\ref{lem-theta}.

Fishers consider the joint distribution $j(\omega,\theta)$ of shocks
to maximize their utility with exogenous insurance contracts in
Equation~\ref{eq-maxrisk}.

\begin{equation}\label{eq-maxrisk}{
\begin{aligned}
U\equiv\max_{x}\mathbb{E}[U]=\int^{\infty}_{-\infty}&\left[ \int^{\bar \omega}_{-\infty}j(\omega,\theta)u(\pi(x,\hat{B},\theta,\omega)+(1-J(\bar \omega))\gamma)d\omega \right.\\
&\left.+\int^{\infty}_{\bar{\omega}}j(\omega,\theta) u(\pi(x,\hat{B},
\theta,\omega)-J(\bar \omega)\gamma)d \omega\right] d\theta
\end{aligned}
}\end{equation}

The first order condition that solves Equation~\ref{eq-max} is then:

\begin{equation}\label{eq-focrisk}
\begin{aligned}
\frac{\partial U}{\partial x}
&= \int_{-\infty}^{\infty}
\Biggl[
\int_{-\infty}^{\bar\omega}
j(\omega,\theta)\,
u_{x}\!\Bigl(
\pi(x,\hat{B},\theta,\omega)
+ (1 - J(\bar\omega))\gamma
\Bigr) \\
&\qquad\qquad\quad
\times
\frac{\partial \pi}{\partial x}
(x,\hat{B},\theta,\omega)
\, d\omega \\
&\quad
+ \int_{\bar\omega}^{\infty}
j(\omega,\theta)\,
u_{x}\!\Bigl(
\pi(x,\hat{B},\theta,\omega)
- J(\bar\omega)\gamma
\Bigr) \\
&\qquad\qquad\quad
\times
\frac{\partial \pi}{\partial x}
(x,\hat{B},\theta,\omega)
\, d\omega
\Biggr]
\, d\theta
= 0 .
\end{aligned}
\end{equation}

Changes in optimal input due to insurance remain characterized by the
implicit function theorem of Equation~\ref{eq-ivt}. Now the numerator is
given by:

\begin{equation}\label{eq-xgamrisk}
\begin{aligned}
\frac{\partial^2 U}{\partial x \, \partial \gamma}
&= \int_{-\infty}^{\infty} 
\Biggl[
\int_{-\infty}^{\bar{\omega}}
j(\omega,\theta)\,
u''\!\Bigl(
\pi(x,\hat{B},\theta,\omega)
+ (1 - J(\bar{\omega}))\gamma
\Bigr) \\
&\qquad\qquad\quad
\times
\frac{\partial \pi}{\partial x}(x,\hat{B},\theta,\omega)
\, (1 - J(\bar{\omega})) \, d\omega \\
&\quad +
\int_{\bar{\omega}}^{\infty}
j(\omega,\theta)\,
u''\!\Bigl(
\pi(x,\hat{B},\theta,\omega)
- J(\bar{\omega})\gamma
\Bigr) \\
&\qquad\qquad\quad
\times
\frac{\partial \pi}{\partial x}(x,\hat{B},\theta,\omega)
\, (- J(\bar{\omega})) \, d\omega
\Biggr]
\, d\theta .
\end{aligned}
\end{equation}

We examine input decisions with insurance contingent on the source of
risk the insurance is designed to protect. Proposition~\ref{prp-ind}
demonstrates that contracts built on extraction risks depend on the
underlying input risk effects.

\begin{proposition}\label{prp-ind}

For index insurance contracts specified at trigger $\bar\omega=0$,
optimal fisher harvest will decrease with an increase in $\gamma$ when
$h_x(x)<0$ and increase when $h_x(x)>0$.

\end{proposition}

\begin{proof}
Independence of $\omega$ and $\theta$ allows us to factor out the
joint distribution in the integral of Equation~\ref{eq-xgamrisk} into
the respective marginal distributions shown by $j_\theta(\theta)$ and
$j_\omega(\omega)$ respectively.

\begin{equation}\label{eq-egam}
\begin{aligned}
\frac{\partial^2 U}{\partial x \, \partial \gamma}
&= \int_{-\infty}^{\infty} j_\theta(\theta)
\Biggl[
\int_{-\infty}^{\bar\omega} j_\omega(\omega)\,
u''\!\Bigl(
\pi(x,\hat{B},\theta,\omega)
+ (1 - J(\bar\omega))\gamma
\Bigr) \\
&\qquad\qquad\quad
\times
\frac{\partial \pi}{\partial x}(x,\hat{B},\theta,\omega)
\, (1 - J(\bar\omega)) \, d\omega \\
&\quad +
\int_{\bar\omega}^{\infty} j_\omega(\omega)\,
u''\!\Bigl(
\pi(x,\hat{B},\theta,\omega)
- J(\bar\omega)\gamma
\Bigr) \\
&\qquad\qquad\quad
\times
\frac{\partial \pi}{\partial x}(x,\hat{B},\theta,\omega)
\, (- J(\bar\omega)) \, d\omega
\Biggr]
\, d\theta .
\end{aligned}
\end{equation}

Suppose insurance fully covers the loss between states, then utility in
the good state and bad state are equal to each other so that we can
factor out like terms in Equation~\ref{eq-egam}. For brevity, all like
terms including $\gamma$ are indicated by $u(\cdot)$.

\begin{equation}\label{eq-simp}
\begin{aligned}
\frac{\partial^2 U}{\partial x \, \partial \gamma}
&= \int_{-\infty}^{\infty} j_\theta(\theta)
\, J(\bar\omega) (1 - J(\bar\omega)) \, u''(\theta,\cdot) \\
&\quad \times
\Biggl[
\int_{-\infty}^{\bar\omega} j_\omega(\omega)\,
\frac{\partial \pi}{\partial x}(x,\hat{B},\theta,\omega)
\, d\omega \\
&\qquad
-
\int_{\bar\omega}^{\infty} j_\omega(\omega)\,
\frac{\partial \pi}{\partial x}(x,\hat{B},\theta,\omega)
\, d\omega
\Biggr]
\, d\theta .
\end{aligned}
\end{equation}

The first term outside the brackets is negative by the definition of
concave utility, $u''<0$. Lemma~\ref{lem-mp} demonstrates the interior
of the brackets is positive when $h_x(x)<0$ as the marginal profit in
the bad state is greater than the marginal profit in the good.
Therefore, index insurance will decrease input use for risk decreasing
inputs when the extraction shocks are independent of stock shocks.

When $h_x(x)>0$, the interior sign of the brackets in
Equation~\ref{eq-simp} is negative by Lemma~\ref{lem-mp}. Therefore,
index insurance will increase input use for risk increasing inputs.

Direction of input change will exactly match the direction of harvest
change by the properties of the production function.
\end{proof}

Proposition~\ref{prp-ind} shows contracts specified on extraction risk
could have positive or negative effects on harvest depending on the risk
effects of the input. Insurance lowers the use of risk decreasing inputs
because the need to protect against risk with that input is replaced by
the risk mitigating qualities of insurance. Insurance increases risk
increasing inputs as it protects against additional risk allowing
fishers to expand production without taking on greater risk. The stock
risk persists when contracts are specified on extraction risks, but is
not influential in the marginal decision with insurance. Instead, if we
specify a contract on the stock shocks, then the stock risk effects
become more prevalent and distinctly changes the harvest outcome.

\begin{proposition}\label{prp-theta}

For index insurance contracts specified at trigger $\bar\theta=0$,
optimal harvest will always increase with an increase in $\gamma$.

\end{proposition}

\begin{proof}
A contract built with $\theta$ will follow the same steps as the proof
for Proposition~\ref{prp-ind} with the only difference being in the
integral bounds and the differential variables as shown in
Equation~\ref{eq-gstheta}. The 2nd term of Equation~\ref{eq-gstheta} is
always negative by Lemma~\ref{lem-theta}. Therefore, a contract built on
$\theta$ will always increase optimal input use.

\begin{equation}\label{eq-gstheta}
\begin{aligned}
\frac{\partial^2 U}{\partial x \, \partial \gamma}
&= \int_{-\infty}^{\infty} j_\omega(\omega)
\, J(\bar\theta) (1 - J(\bar\theta)) \, u''(\omega,\cdot) \\
&\quad \times
\Biggl[
\int_{-\infty}^{\bar\theta} j_\theta(\theta)\,
\frac{\partial \pi}{\partial x}(x,\hat{B},\theta,\omega)
\, d\theta \\
&\qquad
-
\int_{\bar\theta}^{\infty} j_\theta(\theta)\,
\frac{\partial \pi}{\partial x}(x,\hat{B},\theta,\omega)
\, d\theta
\Biggr]
\, d\omega .
\end{aligned}
\end{equation}

The change in input use will lead to higher expected optimal harvest.
\end{proof}

Proposition~\ref{prp-theta} demonstrates a clear bias towards
overfishing with contracts solely specified on $\theta$. Increased
stock abundance directly leads to higher variance. Insurance only
protects against the additional stock risk, and encourages fishers to
expand production. Fishers have no means to lower stock risk through
input choices. The marginal change in variance from insurance cannot be
influenced by extraction risks due to the additive specification.

Insurance contracts can only influence fisher input choices if the
inputs interact with the source of risk the contract is designed to
protect. Choosing which risk to protect has important consequences for
moral hazard effects.

\section{Insurance with multiple inputs}\label{sec-multi}

The single input model provides clear, testable insights. However, real
world fisheries are more complex than single input models. We develop a
multi-input model to represent this complexity. The multi-input model
provides the foundation of our numerical analysis that leverages
parameter estimations from \citet{Asche2020}. Their study
estimated production and risk effect parameters across three inputs in
Norwegian fisheries. The numerical analysis will allow us to quantify
the directional effects of insurance on fisher behavior for valuable
Norwegian fisheries.

We extend the model of the previous section to two inputs,
$X\in\{{x_a,x_b}\}$. Two inputs sufficiently articulate the
complexities that arise while still remaining tractable to solve.

There are two additional effects to consider when adding more inputs.
The interaction between inputs leads to the first effect. Changes in
input use may not correspond to the direction dictated by their
respective extraction risk effects. For example, a fisher may not choose
to reduce a risk decreasing input if the cross partial effects of
production and risk negatively impact production of another input. We
summarize the conditions that lead to unequivocal changes in input use
in Proposition~\ref{prp-samre}.

\begin{proposition}\label{prp-samre}

In fisheries with two inputs, index insurance specified with contracts
on $\omega$ will increase (decrease) the optimal use of a specific
input if the input's risk effects are increasing (decreasing) when the
following sufficient condition is true:

$\frac{\partial U}{\partial x_a\partial x_b}>0$ when both inputs share
the same risk effects, and
$\frac{\partial U}{\partial x_a\partial x_b}<0$ when inputs have
opposite risk effects.

Otherwise, index insurance may lead to ambiguous changes in the input
regardless of the input's own risk effect.

\end{proposition}

The proof is included in Section~\ref{sec-samre}.

The second effect is a straightforward consequence of adding inputs.
While insurance may incentivize greater use of one input, it may
simultaneously reduce use of another. The net change in harvest depends
on both the magnitude of adjustment in each input and their relative
production elasticities.

\begin{proposition}\label{prp-har}

When index insurance leads to increases (decreases) of both inputs,
total expected harvest will increase (decrease).

Otherwise, total change in expected harvest depends on the relative
change in input use and $\frac{\partial f(x_m)}{\partial x_m}$

\end{proposition}

\begin{proof}
The total derivative of expected harvest is:

\begin{equation}\label{eq-totaly}{
\frac{d \mathbb{E}[y]}{dx}=\hat{B}\frac{\partial f(x_a,x_b)}{\partial x_a}dx_a+\hat{B}\frac{\partial f(x_a,x_b)}{\partial x_b}dx_b
}\end{equation}

Marginal production is positive, therefore
$\frac{\partial f(x_m)}{\partial x_m}>0$ for either input represented
by $x_m$. When $dx_a>0$ and $dx_b>0$, Equation~\ref{eq-totaly} is
always positive. The opposite is true when $dx_a<0$ and $dx_b<0$.

For inputs with changes in opposite directions, Equation~\ref{eq-totaly}
is positive or negative contingent on the relative weight between
$\frac{\partial f(x_a,x_b)}{\partial x_a}dx_a$ and
$\frac{\partial f(x_a,x_b)}{\partial x_b}dx_b$
\end{proof}

Proposition~\ref{prp-har} shows that reductions in certain inputs may be
offset by subsequent increases in more productive ones, thereby limiting
the conservation potential of index insurance. In other words, lowering
one margin of production through insurance does not necessarily result
in a smaller total harvest.

These two insights will help explain the modeled responses of fishers in
Section~\ref{sec-sim}. In general, stock and extraction risk effects
remain the leading influences on guiding fishers input choices after
buying insurance. Proposition~\ref{prp-har} and
Proposition~\ref{prp-samre} identify that within the complicated nexus
of multiple input interactions, certain inputs may dominate the overall
outcomes. Both propositions indicate that inputs that share risk effects
(e.g.~all inputs are risk increasing), ought to have the same
conclusions as observed in Section~\ref{sec-common}. Fisheries that use
inputs with opposite risk effects are impossible to sign without further
information. We turn to simulations in the following section to
elucidate the ambiguity.

\section{Numerical Simulations}\label{sec-sim}

We use numerical simulation to determine the magnitude of change in
input use. First, we analyze reasonable parameter estimates to isolate
the magnitude of single input changes. Next, we calibrate the model to
estimates from Norwegian fisheries using three inputs from the
parameters found in\citet{Asche2020}. Monte Carlo simulations
find expected utility across 1000 random draws of stock and extraction
shocks. A comprehensive set of parameters test the sensitivity of fisher
input choices with index insurance. All simulations are conducted in R
with accompanying code available at nggrimes@github.com/ibi-behavior.

\subsection{Simulations with one
input}\label{simulations-with-one-input}

We use the structural form where $f(x)=x^\alpha$ and $h(x)=x^\beta$
to most easily integrate risk increasing or decreasing effects in
$h(x)$. Under these functional forms, Equation~\ref{eq-pi1} becomes:

\begin{equation}\label{eq-sim}{
\pi=x^\alpha(\hat\beta+\theta)+\omega x^\beta-cx^2
}\end{equation}

Mean production $f(x)$ is concave so that $\alpha>0$. Extraction
risk effects on the input can either be risk increasing or decreasing
with $\beta\lessgtr0$. We apply convex costs, $c(x)=cx^2$, for
smoother convergence in the maximization procedure. Stock and extraction
shocks are normally distributed with $\theta\sim N(0,\sigma_{\theta})$
and $\omega\sim N(0,\sigma_{\omega})$.

Fishers will choose inputs $x$ to maximize expected utility with an
exogenous insurance contract. Constant Absolute Risk Aversion (CARA)
utility is used to account for negative shocks and profit loss. Under
this contract specification, the maximization problem becomes
Equation~\ref{eq-maxsim}:

\begin{equation}\label{eq-maxsim}{
\begin{aligned}
U&\equiv\max_{x}\mathbb{E}[u]=\mathbb{E}[(1-\exp(-a(\pi(x,\hat\beta,\theta,\omega)+\mathbb{I}(\gamma))]\\
\mathbb{I}(\gamma)&=\begin{cases}-\rho\gamma & \text{if } \omega\ge \bar\omega\\
(1-\rho)\gamma & \text{if } \omega<\bar\omega
\end{cases}
\end{aligned}
}\end{equation}

We convert $\gamma$ to be a percentage of mean optimal profit without
insurance for interpretability. For example, $\gamma=1$ would
represent a payout equivalent to expected profit before insurance, and
$\gamma=0$ represents no insurance. Equation~\ref{eq-maxsim} shows the
insurance payoff function $I(\gamma)$ built on extraction risk with
$\bar\omega$ as the trigger. A contract built on stock risk would
instead use $\bar\theta$ and $\theta$ as the conditions in the
payoff function.

We create a large parameter space to assess the sensitivity of optimal
input choices to different model parameters. We vary the relative
productivity of the input $\alpha\in\{0.25,0.5,0.75\}$, the extraction
risk effect of the input
$\beta\in\{-0.7,-0.5,-0.3,-0.1,0.1,0.3,0.5,0.7\}$, the risk aversion
parameter $a\in\{1,2,3\}$, the stock shock variance
$\sigma_{\theta}\in\{0.1,0.2,0.3,0.4\}$, and the extraction shock
variance $\sigma_{\omega}\in\{0.1,0.2,0.3,0.4\}$.

For any level of insurance payout $\gamma$, input use changes
monotonically in accordance to the conditions of
Proposition~\ref{prp-theta} and Proposition~\ref{prp-ind}
(Figure~\ref{fig-iter}). Fishers maximize utility at point B in the
bottom right panel of Figure~\ref{fig-iter}. Point A shows the no
insurance case. Any coverage between point A and B will improve utility.
Limits on coverage could constrain fishers to this area, but the
direction of input use remains constant. Overinvestment in insurance, as
shown by point C, could reduce utility as the premiums become excessive
large. Input use continues to increase beyond optimal levels.

The monotonicity of input use in all cases suggests that the insurance
level that maximizes utility will preserve the sign of input changes.
Therefore, an endogenous choice of insurance will not affect the
direction of input change, but it will affect the magnitude.

For example, the equivalent of point B for contracts protecting against
extraction risks when inputs are risk decreasing would be further to
left than contracts on stock risks in Figure~\ref{fig-iter}. Fishers
would be better off with lower levels of insurance coverage. Allowing
fishers to choose insurance coverage ensures that the choice of
insurance is welfare improving and will not bias input choices with over
or under investment of insurance. Simulations moving forward will allow
fishers to choose both inputs and insurance coverage.

\begin{figure}
    \centering
    \includegraphics[width=0.75\linewidth]{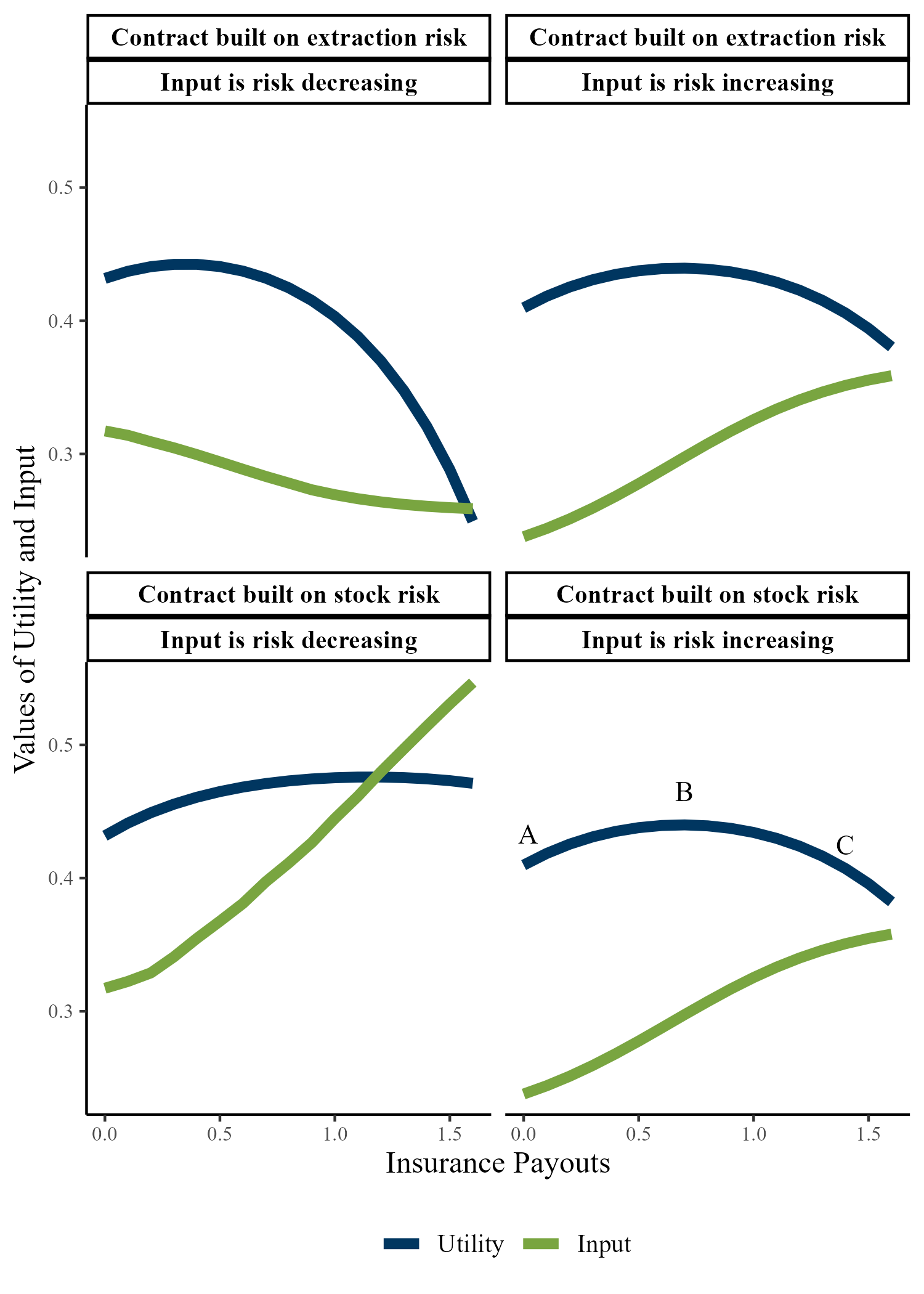}

\caption{\label{fig-iter}Utility (blue lines) and optimal input use
(green lines) with increasing levels of constant index insurance
payouts. Inputs have high mean productivity ($\alpha=0.75$), fishers
are risk aversion $a=3$, and there is high variance in both shocks
($\sigma_{w}=0.4$ and $\sigma_{t}=0.4$).}

\end{figure}%

With two choice variables,$\gamma$ and $x$, Equation~\ref{eq-maxsim}
becomes Equation~\ref{eq-maxsim2}:

\begin{equation}\label{eq-maxsim2}{
\begin{aligned}
U&\equiv\max_{x,\gamma}\mathbb{E}[u]=\mathbb{E}[(1-\exp(-a(\pi(x,\hat\beta,\theta,\omega)+\mathbb{I}(\gamma))]\\
\mathbb{I}(\gamma)&=\begin{cases}-\rho\gamma & \text{if } \omega\ge \bar\omega\\
(1-\rho)\gamma & \text{if } \omega<\bar\omega
\end{cases}
\end{aligned}
}\end{equation}

Furthermore, we run two groups of simulations: (i) one where the
insurance contract indemnifies on $\omega$, as in
Equation~\ref{eq-maxsim}; and (ii) one where the index is constructed on
$\theta$, to present results for both Proposition~\ref{prp-ind} and
Proposition~\ref{prp-theta}.

Contract specification and risk effects control the direction of input
change in fisheries. Insurance contracts that protect against extraction
risks lead to increases or decreases in optimal input use contingent on
the underlying risk effects of the input. All optimal choices of risk
decreasing inputs, shown as red bars in Figure~\ref{fig-corr}, decreased
if the contract was specified on extraction risks.

The productivity of inputs strongly influences the magnitude of change
in input use particularly for risk decreasing inputs. When inputs are
relatively less productive (left panel, $\alpha=0.25$), fishers are
more willing to reduce the unproductive input in favor of the protection
offered by insurance. They lose less in production while gaining more
variance reduction by substituting with insurance. Therefore an
important tradeoff exists for risk reducing inputs that does not exist
for risk increasing inputs. Fishers demonstrate smaller absolute changes
in input use for risk decreasing inputs than risk increasing
particularly at higher mean productivity levels in
Figure~\ref{fig-corr}. Risk increasing inputs exert stronger effects
towards overfishing than risk decreasing inputs have towards
conservation.

\begin{figure}
    \centering
    \includegraphics[width=0.75\linewidth]{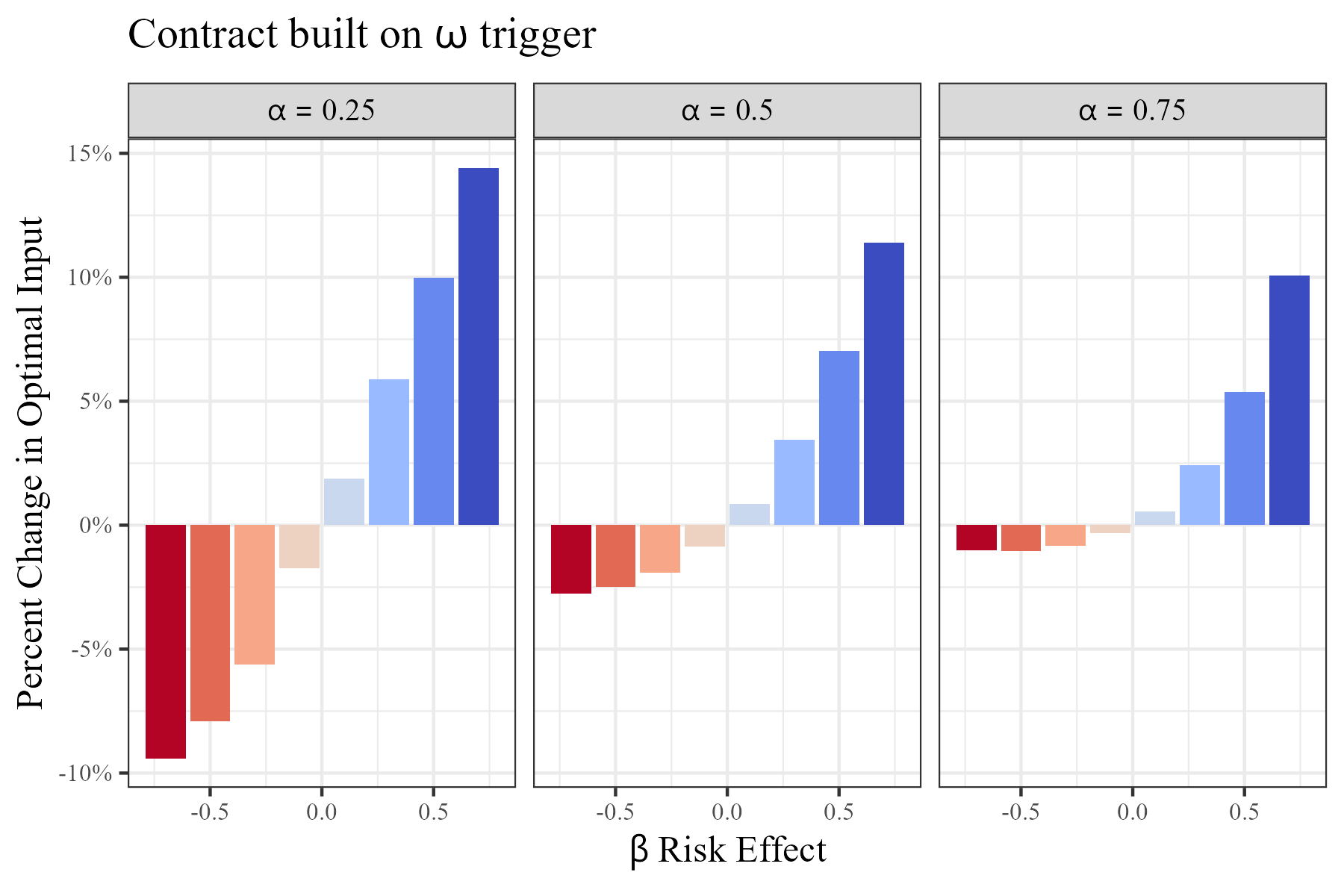}

\caption{\label{fig-corr}Percentage change in optimal input with an
index insurance contract using extraction risk, $\omega$, as the
index. Risk increasing inputs (blue bars) always increase input use,
while risk decreasing inputs (red bars) always decrease input use. Each
panel indicates the mean productivity ($\alpha$) of the input.}

\end{figure}%

Contracts built on $\theta$ as the index always lead towards
overfishing because of the inherent risk increasing characteristics of
$f(x)$ (Figure~\ref{fig-corr-theta}). Inputs that are more productive
and more risky with higher levels of $\alpha$ will be used more. The
right most panel of Figure~\ref{fig-corr-theta} shows the highest effect
where input changes could increase by 18\% relative to no insurance.
Extraction risk in $h(x)$ has little influence on changes in optimal
input use when the contract is specified on stock risk. Across all
levels of the risk effect, $\beta$, the magnitude of change remains
relatively constant for a given productivity level.

\begin{figure}
    \centering
    \includegraphics[width=0.75\linewidth]{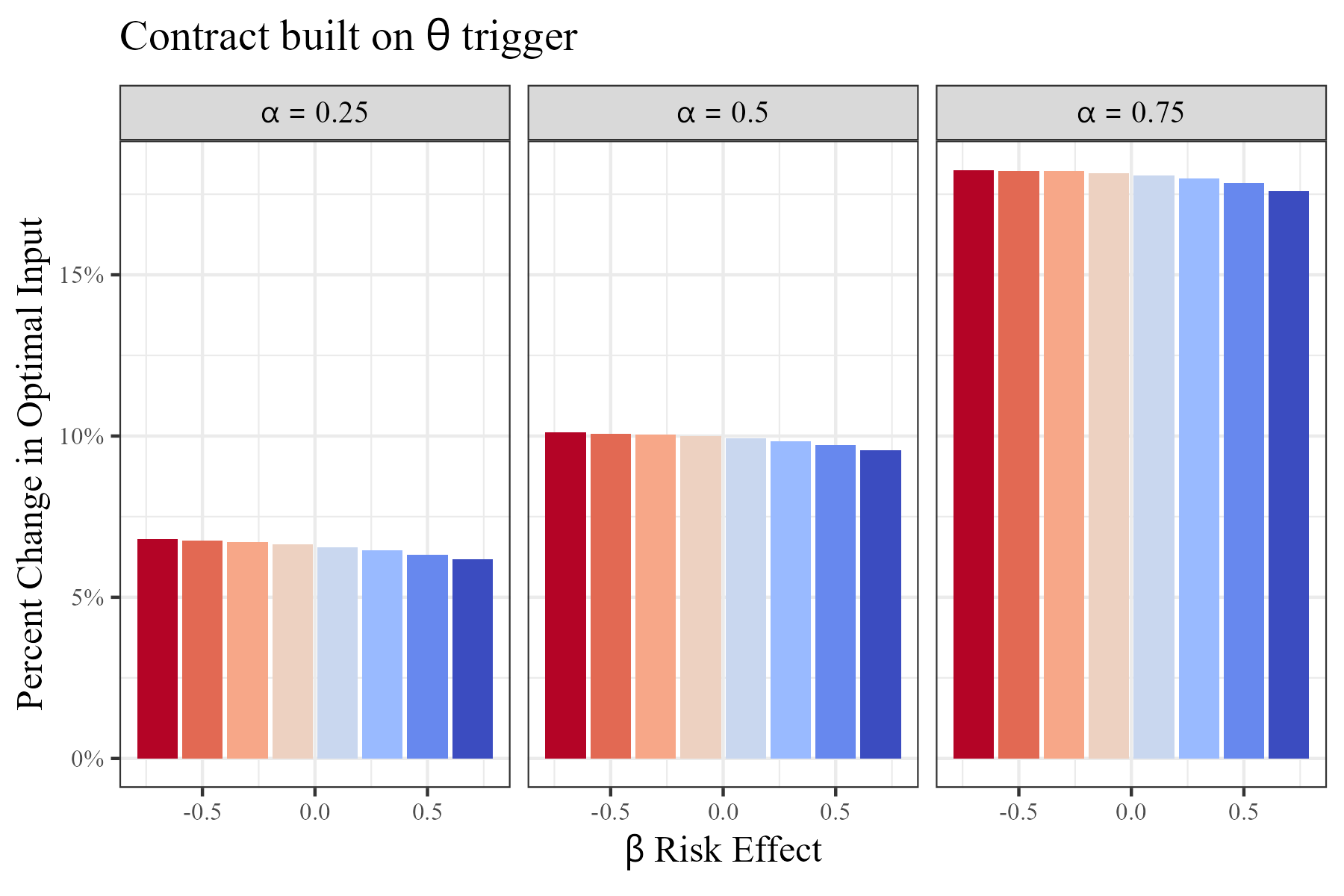}
\caption{\label{fig-corr-theta}Percentage change in optimal input with
an index insurance contract using stock risk, $\theta$, as the index.
Risk increasing inputs (blue bars) and risk decreasing inputs (red bars)
always increase input use. Each panel indicates the mean productivity
($\alpha$) of the input.}

\end{figure}%

Magnitude of input changes are sensitive to other parameters beyond just
mean production elasticity ($\alpha$). We quickly demonstrate some
comparative statics of other important variables such as risk aversion
(Panel A in Figure~\ref{fig-sum}), trigger thresholds (Panel B),
variance of the extraction risk $\sigma_\omega$ (Panel C), and
variance of the stock risk $\sigma_\theta$ (Panel D). In each panel,
we show the results for the contracts specified on stock risk
($\theta$) and extraction risk ($\omega$). The x-axis of each panel
is the extraction risk effect coefficient $\beta$, while the y-axis is
the percent change in input use with index insurance compared to no
insurance. Different colors represent the comparative statics of each
new parameter. Across all panels, the new parameters affect only the
magnitude of change. The direction of change continues to be determined
by risk effects and contract specification.

More risk averse fishers respond more aggressively to insurance (Panel A
in Figure~\ref{fig-sum}). Risk aversion implies more sensitivity towards
risk. The protection from insurance has greater marginal value for more
risk averse fishers. Therefore, they adjust their input use more
significantly to take advantage of the risk mitigating qualities of
insurance.

Different triggers also lead to different magnitudes of input change,
but do not affect the sign changes. While necessary for applying
Lemma~\ref{lem-mp} and Lemma~\ref{lem-theta} in the proofs, the results
of Proposition~\ref{prp-ind} and Proposition~\ref{prp-theta} would
appear to hold if $\bar\omega\ne0$ and $\bar\theta\ne0$. Higher
trigger levels protect against extreme or catastrophic shocks. However,
those shocks occur much less frequently than those closer to the mean.
Fishers are less incentivized to change their expected production levels
in the event of rare events. Offering different insurance contracts does
change the optimal choice of insurance in line with the results of
\citet{Lichtenberg2022}.

Fisher input choice is much more responsive to insurance protection from
more variable shocks (Panel C and D Figure~\ref{fig-sum}), but it
depends on the insurance contract. Similar to risk aversion, the greater
the shocks the greater the marginal value of insurance to mitigate those
shocks. In more volatile environments, insurance provides significantly
more income smoothing leading to similar incentives as the higher risk
aversion example. Shocks the insurance does not protect against has
little influence on fisher input decisions. If the insurance is designed
to protect against one specific shock, it is unsurprising changes in the
other independent shock does not affect fishers decisions.

\begin{figure}
    \centering
    \includegraphics[width=1\linewidth]{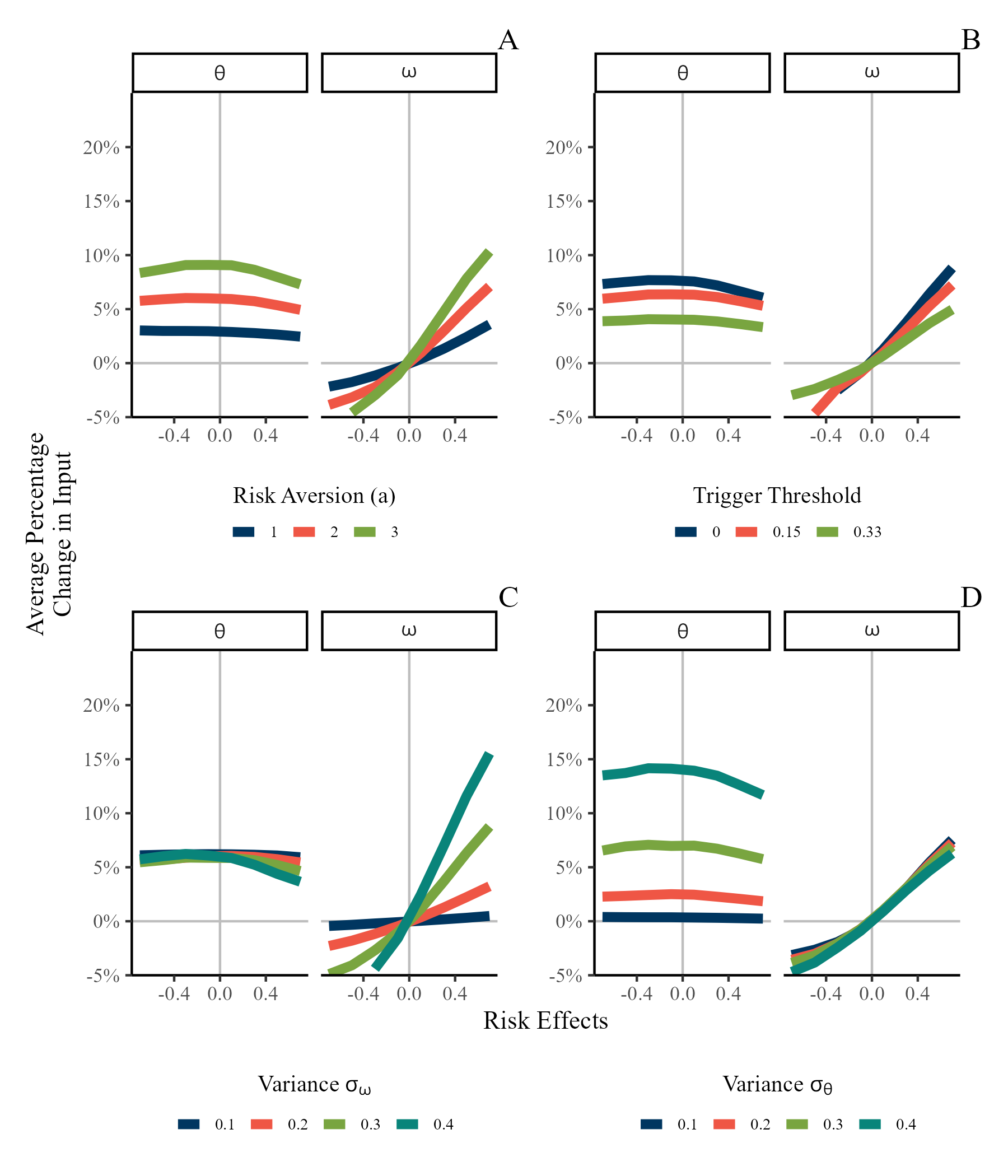}

\caption{\label{fig-sum}Risk Aversion (A), trigger threshold (B), stock
variance $\sigma_{\omega}$ (C), and extraction variance
$\sigma_{\theta}$ (D) all influence the magnitude of change in input
use. Mean production elasticity is set to 0.5. Average percent change in
input (y-axis) is summarized across all other parameter combinations for
each risk effect value of $\beta$. Contracts built on extraction risk
are in the subpanels with $\omega$, while contracts built on stock
risk are indicated by the $\theta$ subpanel.}

\end{figure}%

\subsection{Application to Norwegian
Fisheries}\label{application-to-norwegian-fisheries}

Fishery extraction risk effects are rarely estimated, though appear
crucial to determining the magnitude of input change with index
insurance. The only study to date that calculates risk effects is that
of \citet{Asche2020}. They used a non-linear estimator to
calculate the production and risk parameters of a Just-Pope function
across four different vessel types. We will use their coefficient
calculations to calibrate an estimate of the magnitude of input and
harvest change that index insurance would incentivize if offered to
Norwegian fisheries.

\citet{Asche2020} aggregated by vessel type and not species, so
there is no reasonable estimate for biomass. They accounted for biomass
using fixed effects in their regression, but without additional
information we cannot parameterize the mean and variance of biomass.
Therefore, our simulations normalize mean biomass to 1 and we assume the
stock shocks, $\theta$, have a normal distribution and test different
levels of variance. Norwegian fisheries are well managed so the stock
variance could be mitigated through quota systems or accurate stock
assessments. The simulation model uses three inputs: capital $k$,
labor $l$, and fuel $f$, yielding the following profit in
Equation~\ref{eq-sim3}:

\begin{equation}\label{eq-sim3}{
\pi(k,l,f)=k^{\alpha_k}l^{\alpha_l}f^{\alpha_f}(\hat\beta+\theta)+\omega k^{\beta_k}l^{\beta_l}k^{\beta_f}-c_kk^2-c_ll^2-c_ff^2
}\end{equation}

Mean production, $\alpha$, and risk, $\beta$, elasticities control
the stock and extraction risk effects respectively. Each parameter is
indexed to a particular input through the subscript, e.g.~fuel mean
production elasticity is $\alpha_f$. Fishers in the simulation choose
inputs and insurance coverage to maximize expected utility. We show
their choice based on a $\omega$ contract
(Equation~\ref{eq-maxasche}), but also run a model specification with a
contract built on $\theta$.

\begin{equation}\label{eq-maxasche}{
\begin{aligned}
U&\equiv\max_{\gamma,k,l,f}\mathbb{E}[u]=\mathbb{E}[u(k^{\alpha_k}l^{\alpha_l}f^{\alpha_f}(\hat\beta+\theta)+\omega k^{\beta_k}l^{\beta_l}k^{\beta_f}-c_kk^2-c_ll^2-c_ff^2+\mathbb{I}(\gamma)]\\
\mathbb{I}(\gamma)&=\begin{cases}-\rho\gamma & \text{if } \omega\ge \bar \omega\\
(1-\rho)\omega& \text{if } \omega<\bar \omega
\end{cases}
\end{aligned}
}\end{equation}

Table~\ref{tbl-asche} shows the production and risk elasticities of the
three vessel types found in Norway\footnote{In \citet{Asche2020}, they also estimate a fourth vessel type, purse seiners.
  However the point estimate for mean production elasticity for labor
  was negative which violates our assumption of increasing $f(x)$.
  Therefore we drop purse seiners from our analysis.}.

\begin{table}[htbp]
\centering
\caption{\label{tbl-asche}Production and Risk elasticities of Norwegian Fisheries from \citet{Asche2020}}
\begin{tabular}{lcccccc}
\hline
 & $\alpha_k$ & $\alpha_l$ & $\alpha_f$ & $\beta_k$ & $\beta_l$ & $\beta_f$ \\
\hline
Coastal Seiners      & 0.294 & 0.421 & 0.457 & 0.184  & -0.432 & 0.119 \\
Coastal Groundfish   & 0.463 & 0.421 & 0.355 & 0.965  & -0.080 & 0.113 \\
Groundfish Trawlers  & 0.210 & 0.106 & 0.531 & -2.788 & -0.110 & -0.024 \\
\hline
\end{tabular}
\end{table}

We use the same parameter space as the previous simulations to test the
sensitivity of fisher input choices with index insurance. We plot the
distribution of input change after insurance for all combination of
parameters in Figure~\ref{fig-asche-input} based on a contract with
$\omega$ as the index. Figure~\ref{fig-asche-input} also allow us to
examine whether the conditions of Proposition~\ref{prp-samre} hold with
real world combinations. We report the median change as an indicator of
the general direction of input change.
\begin{figure}
    \centering
    \includegraphics[width=1\linewidth]{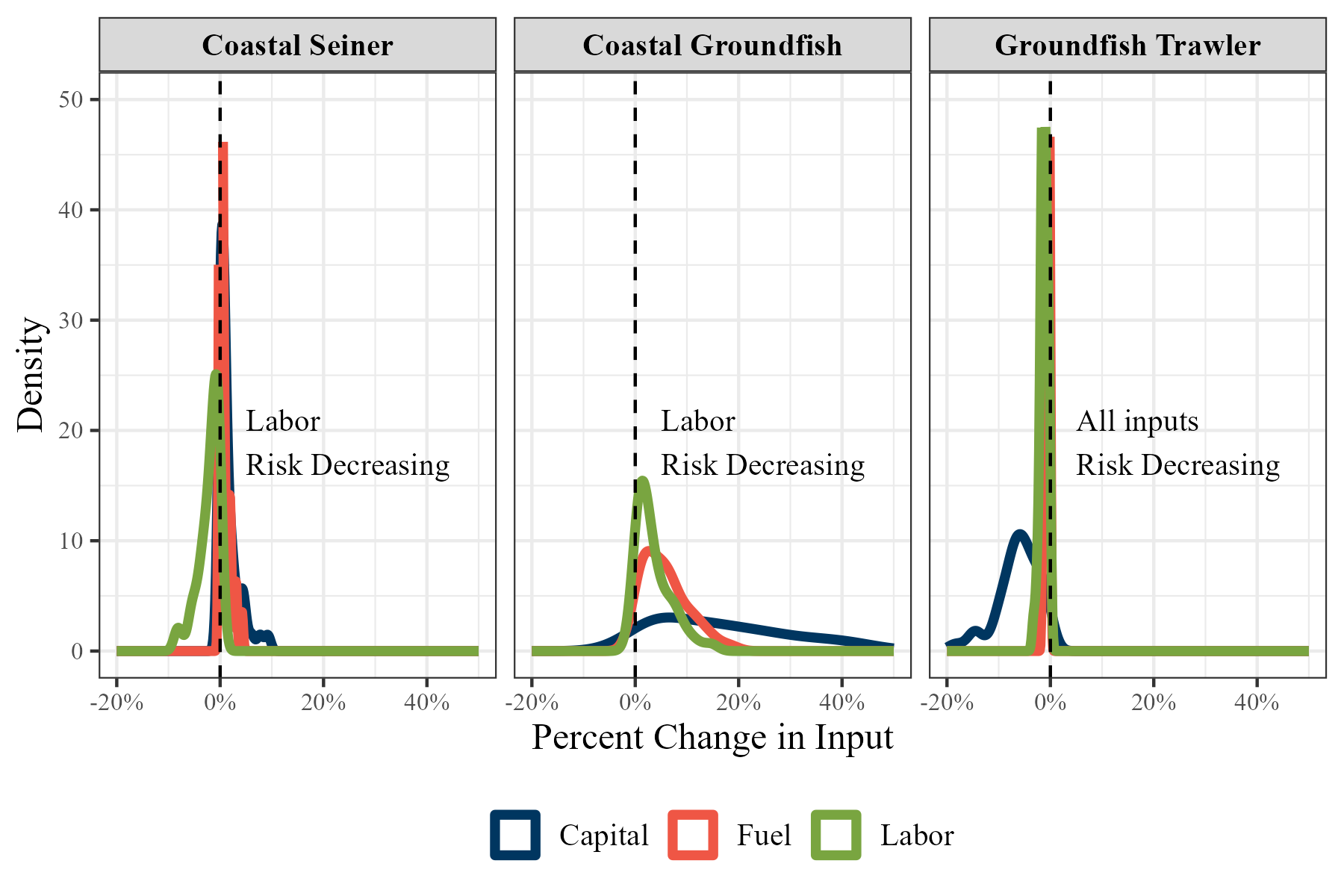}

\caption{\label{fig-asche-input}Density plots of the percent change in
input use for each vessel type in Norwegian fisheries with contracts
built on extraction risk $\omega$. The dashed black line represents no
change in input use. Risk decreasing inputs are labeled.}

\end{figure}

Most changes in inputs tend to change in the direction expected of their
own individual risk effects. For example, fuel and capital are risk
increasing inputs for Coastal Seiners and saw small, but positive
increases. Labor was risk decreasing and saw a decline in use. This
indicates that the conditions of Proposition~\ref{prp-samre} can hold.
The relatively even balance between the marginal productivity
elasticities and risk elasticities perhaps determine the conditions.

Labor in the Coastal Groundfish fishery is risk decreasing, yet always
saw an increase in the simulations. The risk parameter of labor is
relatively small compared to the risk increasing coefficients of capital
and fuel. The cross partial mix of the inputs may explain why fishers
add more labor. As insurance strongly incentivizes capital and fuel
increases, labor must also increase to further enhance those other
inputs. Here, the conditions of Proposition~\ref{prp-samre} do not hold.

Every input in the groundfish trawler fishery was risk decreasing. With
contracts on extraction risks all inputs saw a decline. The largest
decline was in capital, because capital possesses the strongest risk
decreasing effect out of all inputs.

Recreating the same analysis with contracts on stock risks shows the
propensity for stock risk contracts to stimulate increased input use
regardless of risk effects (Figure~\ref{fig-asche-all}). All inputs saw
similar increases in each fishery with the exception of capital in the
groundfish trawlers. While capital still increased, it did so at a much
smaller rate. Capital in the groundfish trawler fishery is the most risk
decreasing input out of all the fisheries. Fishers do not need to expand
capital production in this case because it already reduces so much risk.
Adding insurance marginally reduces some risk, hence the increase, but
is not needed as much compared to other fisheries with weaker risk
reducing inputs.

\begin{figure}
    \centering
    \includegraphics[width=1\linewidth]{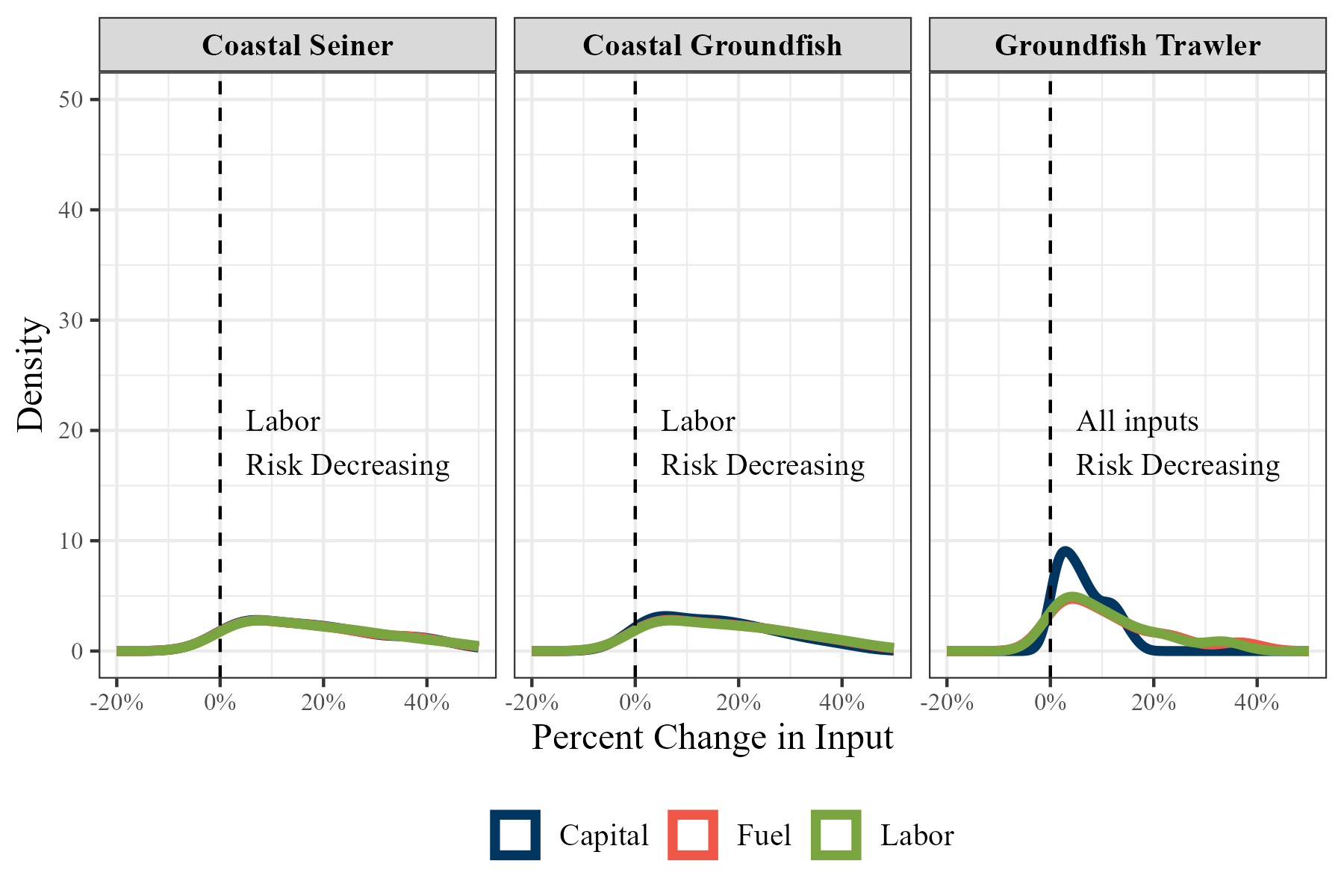}

\caption{\label{fig-asche-all}Density plots of the percent change in
input use for each vessel type in Norwegian fisheries with insurance
contracts built on stock risk $\theta$. The dashed black line
represents no change in input use. Risk decreasing inputs are labeled.}

\end{figure}

Input changes lead to harvest changes. We examine the total change in
harvest for all parameters in Figure~\ref{fig-asche} for a contract
triggered on $\omega$. Overall, insurance leads to relatively small
changes in harvest for all fisheries, but increases are stronger than
decreases. Coastal Groundfish see the largest and most consistent
increase in harvest. Median harvest increased by 10\% with a max
increase of 36\%. Coastal Groundfish have the most risk increasing
inputs out of the estimated fisheries, and had greater input use
(Figure~\ref{fig-asche-input}).

\begin{figure}
    \centering
    \includegraphics[width=1\linewidth]{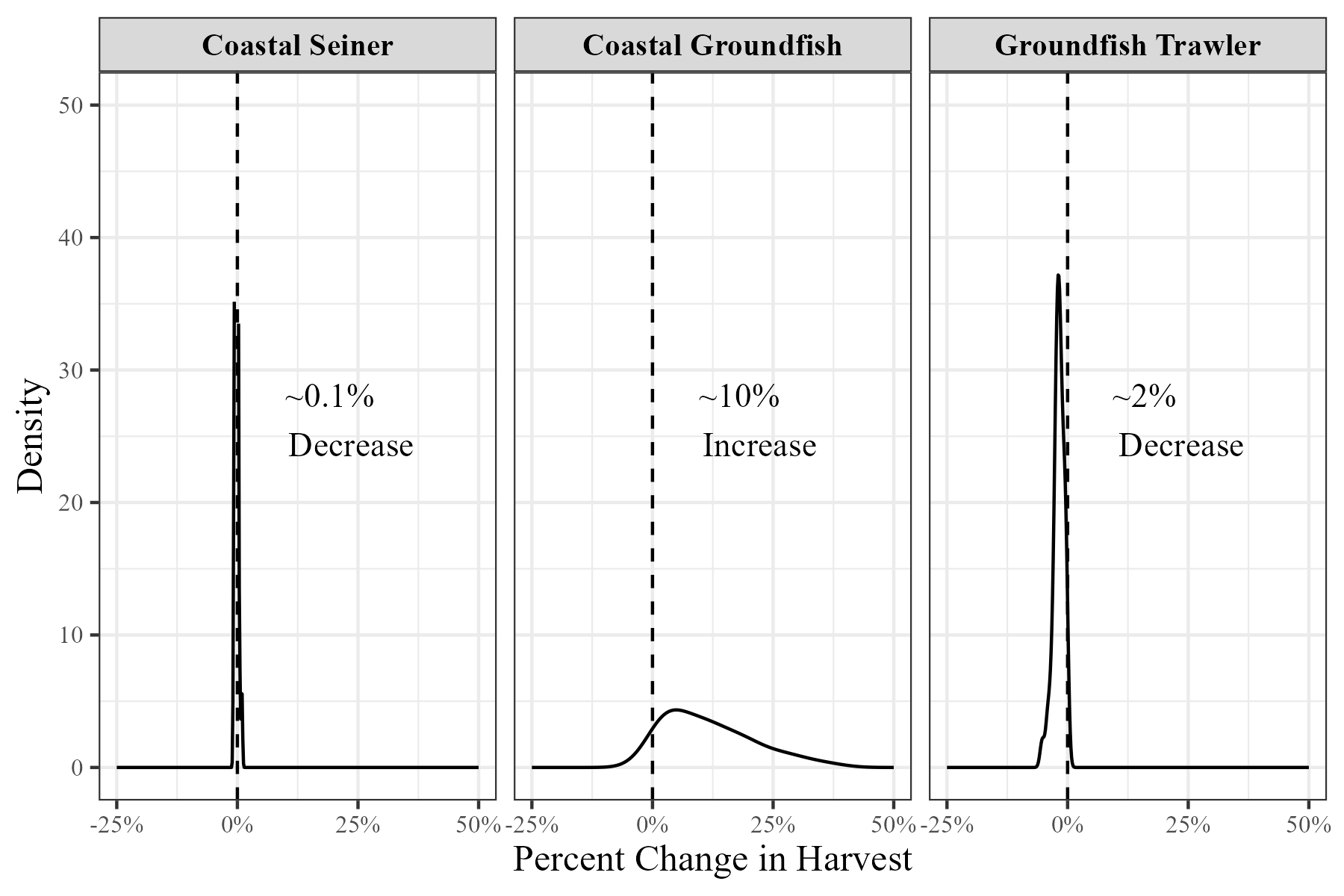}

\caption{\label{fig-asche}Density plots of the percent change in harvest
for each vessel type in Norwegian fisheries from an insurance contract
trigged on extraction risks $\omega$. The dashed line represents no
change in harvest. The text labels represent the median percent change
in harvest for each vessel type.}

\end{figure}

Coastal Seiners had a relatively balanced spectrum of risk effects. The
input mix in this case led to both increases and decreases in input use,
which on net led to near zero changes in harvest. There is a slight skew
towards increased harvest, but drastically less than the Coastal
Groundfish fishery.

Groundfish trawlers consistently see small decreases of 2\% in harvest
(Figure~\ref{fig-asche}). Capital was the primary driver of harvest
reduction. However, it has a relatively low marginal productivity.
Insurance decreases trawler capital use by about 8\%, but the low
productivity leads to only a 2\% decrease in overall harvest.

Applying an insurance contract indemnified on $\theta$ instead of
$\omega$ shifts the direction towards more overfishing
(Figure~\ref{fig-asche-theta}). Prominent shifts occur in all fleets.
For example, in Figure~\ref{fig-asche}, the percent change in harvest
for Coastal Seiners is indistinguishable from zero. With a $\theta$
index contract, Coastal Seiners would increase harvest by 18\%
(Figure~\ref{fig-asche-theta}).

Despite the risk decreasing dominance of capital in groundfish trawlers,
fishers will choose to increase production as insurance protects against
the added harvest risk. This result most clearly shows the impact of
different insurance contracts and the potential for maladaptive behavior
change. Without considering all the margins for change, insurance
protecting against biological risk will encourage overfishing without
additional constraints.

\begin{figure}
    \centering
    \includegraphics[width=1\linewidth]{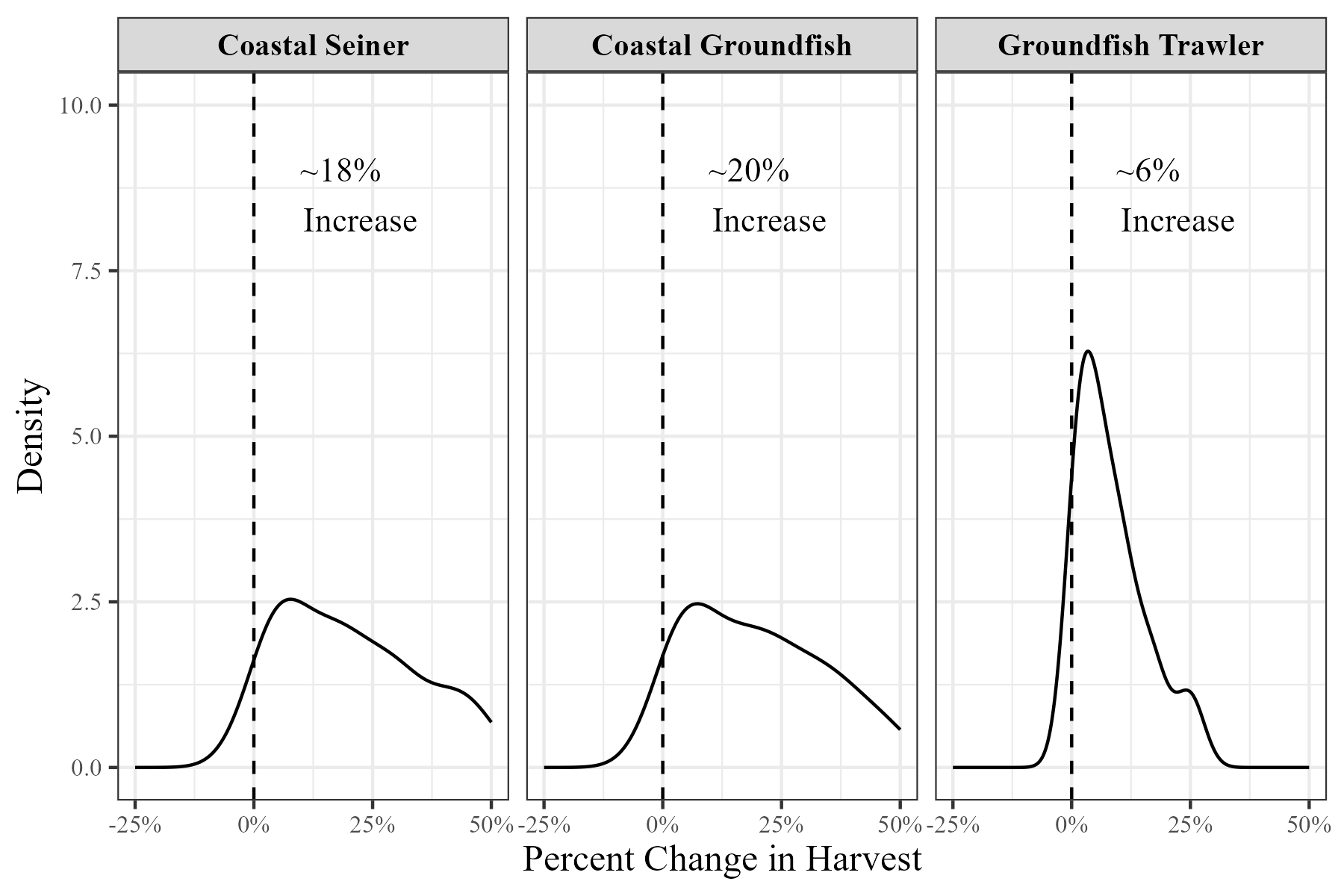}
\caption{\label{fig-asche-theta}Density plots of the percent change in
harvest for each vessel type in Norwegian fisheries with insurance
contract indemnified on biological risk $\theta$. The dashed line
represents no change in harvest. The text labels represent the median
percent change in harvest for each vessel type.}

\end{figure}

\section{Discussion}\label{sec-disc}

This paper makes four distinct contributions. First, if fishery
production is specified to traditional fishery models, then index
insurance will lead to overfishing. Second, contract specification
influences fisher behavior. Contracts that use stock risk as the index
will always lead to increases, while contracts that use extraction
shocks as the index could increase or decrease. The risk effects are the
key drivers of this result similar to prior findings in agriculture
\citep{horowitz1993,Ramaswami1993,Mahul2001}. Third,
multiple inputs make it more difficult to predict individual changes in
input use and overall changes in harvest. The overall effect of index
insurance on harvest requires the total effect of all input changes.
Fourth, insurance incentives that increase harvest lead to larger
absolute changes than decreases in harvest.

The fundamental driver of fishers' behavior changes is whether the
marginal change in productivity is balanced by the marginal change in
risk. Fishers are more willing to increase production if insurance
negates the additional risk of expanded production. Since insurance
lowers risk, fishers need less self insurance through risk reducing
inputs and can reduce their overall input use. However, using less
inputs implies less catch and revenue creating a unique tension for risk
decreasing inputs. Across all simulations, decreased input use was
smaller than increased input use holding all other parameters constant.
This is a clear demonstration of the tradeoff fishers face when
considering risk decreasing input choices with insurance. Behavior
change in fisheries will lean towards expanding production creating a
dilemma for conservation efforts.

Index insurance would improve welfare in Norwegian fisheries, but also
lead to changes in harvest that depends on the extraction risk effects
of fishing fleets. The average utility gain was 2\% on
average\footnote{The utility gains are smaller by construction as we
  manually introduce basis risk with an imperfect insurance contract
  that cannot simultaneously protect both risks.}, indicating index
insurance would be welfare improving. Therefore, insurance can provide
immediate benefit to fishers, but the policymaker must take measures to
mitigate the preserve incentive to expand fishing production. Otherwise,
the long term health of the stock could be degraded and fishers would be
worse off in the long run \citep{muller2017,John2019,Bulte2021}.

Norwegian pelagic groundfish trawlers would have the opposite
considerations. When offered insurance they reduced their harvest by
2\%. Though small, it will lead to improved fishery sustainability. The
long term benefit of insurance would increase with improved stock
health. The decline in harvest was driven by a reduction in
overcapitalization, because capital was a significant risk decreasing
input. Policymakers should attempt to identify fisheries with risk
decreasing inputs for insurance contracts to improve sustainability if
index insurance is to operate in isolation of other policies.

Ex-ante identification of input risk effects is challenging. Extraction
risk effects remain an elusive concept in fisheries, and need to be
estimated in order to articulate more accurate behavior changes of
fishery index insurance. Crop covers and pesticide provide clear
examples of risk decreasing inputs in agriculture, but what do risk
decreasing inputs look like in fisheries? Labor is perhaps the more
intuitive risk decreasing input. Technical expertise of crew and
captains can hedge against luck when fishing \citep{Alvarez2006}. Better trained crew can deploy gear in a safe and timely manner,
increasing the likelihood of effective sets.

Fuel as a risk increasing input in fisheries makes intuitive sense as
well. Fuel is used to power vessels and is a direct cost of fishing.
Fishers explore productive fishing grounds for the best location. Every
hour at sea increases the harvest reward, but also the chances of
failure.

Capital is a more complex input, because it is shown to be both risk
increasing and decreasing. Capital in fisheries typically refer to
vessel tonnage, engine power, and gear technology. \citet{Asche2020} noted that the trawler fishery, $\beta_k=-2.7$, was primarily
an ocean going fleet that could move to the fish and handle adverse
weather more readily. The authors postulated the flexibility to pursue
fish reduced overall production variance. Quota allocations in the
coastal fisheries are determined by vessel size with smaller vessels
receiving proportionally more quota. Increasing vessel size would
increase risk by introducing more variable quotas. The context of the
fishery likely determines the risk effects of capital.

Capital acting as a risk decreasing input may explain overcapitalization
in many fisheries around the world. In standard common pool resource
models, adding risk aversion ought to lower aggregate effort
\citep{gibbons1993,Tilman2018}. Yet,
overcapitalization and overfishing are more often observed in the real
world. Either fishers are never risk averse or the risk effects of
capital are not as simple as the standard model suggests. When capital
is allowed to be risk decreasing, optimal capital choices are much
higher than without risk. Thus, fishers can still make rational, risk
averse decisions even while overfishing.

Specifying contracts on stock risk or extraction risk has clear
differences. However, the leading candidates for possible indices in
fisheries index insurance are currently weather variables most often
associated with biological stock risks \citep{Watson2023}.
Designing contracts solely on these variables may lead to harvest
increases that run contrary to conservation goals.

The structural limitations imposed by assumptions on stock risk may
drive this result. Most bioeconomic models simplify the complex effects
of stock dynamics into multiplicative or additive forms as modeled in
this paper. Instead, different forms of risk could be embedded into the
biological component of fishery models. Stock variance could be greater
in overfished stocks instead of healthier ones, reflecting more
vulnerability in weaker states \citep{Sims2018}. Adapting
alternative, more biologically focused specifications of stock risk
could change the behavioral effects of insurance. Fishers may be more
willing to expose themselves to greater risk at more vulnerable stock
levels with insurance. Alternatively, insurance could help mitigate risk
and incentivize fishers to move toward healthier stocks with less
variance by alleviating income pressures to fish. Further analysis is
required to understand the full implications of stock risk effects in
fisheries.

Management may also limit stock risk, and interact with insurance in
complex ways. Our analysis explicitly modeled scenarios without the
existence of management. We wanted to analyze the interaction of
insurance on fisher behavior in unconstrained settings first to derive a
clearer incentive structure. Most fisheries are managed in some form.
The interaction between management and insurance may be complementary or
substitutes. For example, well managed fisheries that have responsive
harvest control rules may not need insurance. Stock risk could be
reduced so that only contracts on extraction risks would be viable.
Insurance demand and uptake may be low in these fisheries.

Insurance could instead complement management to provide the financial
relief that management cannot offer. Managers often focus on the
biological health of the fishery that can run at odds with fishers'
desires to enhance their income. Insurance can act as the financial
relief to allow managers to pursue more active strategies to protect
fish stocks without political resistance from lowered quotas.
Additionally, management can provide the constraints on insurance moral
hazard so the income smoothing benefits are passed to fishers, but not
the long term degradation. The interaction between insurance and
management requires further investigation especially with the numerous
management strategies that exist in fisheries.

The transfer between inputs and insurance reflects the substitution
between self-insurance and formal insurance \citep{Quaas2008}. If index insurance is designed to reduce fishing capacity,
efforts must be made to ensure that it does not take away from the self
insurance capacity of fishers. Labor appears to be consistently risk
reducing and acts as a form of self insurance. If index insurance
incentivizes captains to hire less crew, the stock of fish may be
preserved, but less employment may reverberate throughout the community.
Fishing is often a primary employment opportunity in coastal
communities. The resiliency of the community would be compromised rather
than enhanced with fewer jobs. Contract stipulations could mandate that
only cost expenses are covered by payouts thereby including lost wages
to the crew. Agriculture contracts often are designed to directly cover
expenses \citep{He2020}.

Our model only directly models behavior change through moral hazards.
Index insurance could be designed to incentivize other forms of
sustainable behavior change. We define three pathways insurance can
change behavior: Moral hazards, Quid Pro Quo, and Collective Action.
Moral hazards were proven in this paper to have ambiguous impacts
controlled by the risk characteristics of fishery inputs. The incentives
of moral hazards will always exist, therefore other measures could be
taken to either limit the downside behavioral effects of insurance or
stimulate other forms of sustainable behavior.

Quid Pro Quo expands contract design to explicitly build in conservation
measures. Fishers would be required to adopt sustainable practices in
order to qualify for insurance. Quid Pro Quo is already used in
agricultural insurance in the form of Good Farm Practices. Farmers must
submit management plans to US Risk Management Agency that clearly
outline their conservation practices in order to qualify for insurance.
Working closely with management agencies, insurance companies could
design contracts that require fishers to follow fishery specific
management practices. For example, fishers may be incentivized to use
more sustainable gear types, have an observer onboard, or reduce
bycatch. A stipulation in in the COAST program is fishers must register
their vessels with the participating countries fisheries department
\citep{Sainsbury2019}. This requirement has brought greater
data clarity to small scale Caribbean fisheries.

Further research would need to uncover the full impact of Quid Pro Quo,
but an initial hypothesis would be fishers will be willing to adopt
sustainable practices so long as the marginal gain in utility from the
insurance is greater or equal to the marginal costs of the stipulated
sustainable changes. Otherwise fishers will not want to buy the
contracts, and the insurance has no binding constraints to change the
fishery.

Collective action ties insurance premiums to biological outcomes to
leverage the political economy of the fishery. Insures could reduce
premiums in fisheries that have robust management practices such as
adaptive harvest control rules, stock assessments, or marine protected
areas in the vicinity. Fishers could either pressure regulators to adopt
these actions or form industry groups to undertake the required actions.
Insurers would agree to this if triggers are connected to biological
health so that negative shocks are less frequent and thus payouts occur
less. Fishers gain from the reduced insurance premium and the increased
sustainability of harvest with rigorous management in place.

Ultimately, if index insurance is to be used in fisheries, it must be
designed with clear objectives and intentions. Index insurance can meet
objectives of income stability and risk reduction, but there has been an
implicit assumption by practitioners that index insurance will always
lead to improved sustainability. Our results directly contradict that
claim with the use of traditional models. Without considering the
behavior change of fishers when adopting insurance, the outcomes may not
be as expected. New insights derived from this paper will help guide the
efficient and sustainable implementation of fisheries index insurance.

\newpage
\appendix
\renewcommand{\thefigure}{A\arabic{figure}}
\renewcommand{\thetable}{A\arabic{table}}
\setcounter{figure}{0}
\setcounter{table}{0}

\section{Appendix}\label{appendix}

\subsection{\texorpdfstring{Proof of
Lemma~\ref{lem-theta}}{Proof of Lemma~}}\label{proof-of-lem-theta}

\textbf{Lemma 2.1} \emph{Index insurance contracts built on} $\theta$
\emph{will always lead to higher expected marginal profits in the good
state}

$\frac{\mathbb{E}[\partial \pi|\theta<\bar \theta]}{\partial x}-\frac{\mathbb{E}[\partial \pi|\theta>\bar \theta]}{\partial x}<0$

\begin{proof}
With standard fishery production, $\pi=f(x)(\hat\beta+\theta)-c(x)$,
the first order conditions lead to optimal input choices that are equal
across states of the world. The choice of inputs occurs before the realization of the stochastic shock. Therefore $f(x^*)$, and $c(x^*)$ are
equal across states. Where $x^*$ denotes the optimal input choice.

\begin{equation}\label{eq-compz1}{
\begin{aligned}
\scriptstyle
\frac{\partial \mathbb{E}[\pi|\theta<\bar\theta]}{\partial x^*}-&\frac{\partial \mathbb{E}[\pi|\theta>\bar\theta]}{\partial x^*}=\\&\cancel{f_{x^*}(x^*)\hat{B}}+\mathbb{E}[\theta f_x(x^*)|\theta <\bar\theta]-\cancel{c_{x^*}(x^*)} \\
&-\cancel{f_{x^*}(x^*)\hat{B}}-\mathbb{E}[\theta f_x(x^*)|\theta >\bar\theta]+\cancel{c_{x^*}(x^*)}\\
=&\mathbb{E}[\theta f_x(x^*)|\theta <\bar\theta]-\mathbb{E}[\theta f_x(x^*)|\theta >\bar\theta]
\end{aligned}
}\end{equation}

With risky production profits are
$\pi=\omega h(x)+f(x)(\hat\beta+\theta)-c(x)$. Independent shocks lead
$\mathbb{E}[\omega|\theta \lessgtr \bar \theta]=0$.

\begin{equation}\label{eq-compz2}
\begin{aligned}
\frac{\partial \mathbb{E}[\pi \mid \theta < \bar\theta]}{\partial x^*}
-
\frac{\partial \mathbb{E}[\pi \mid \theta > \bar\theta]}{\partial x^*}
&=
\mathbb{E}\!\left[
\theta f_x(x^*)
\,\middle|\, \theta < \bar\theta
\right]
-
\mathbb{E}\!\left[
\theta f_x(x^*)
\,\middle|\, \theta > \bar\theta
\right] \\
&\quad
+\, \cancel{f_{x^*}(x^*) \hat{B}}
- \cancel{c_{x^*}(x^*)} \\
&\quad
+\, \cancel{
\mathbb{E}\!\left[
\omega h_{x^*}(x^*)
\,\middle|\, \theta < \bar\theta
\right]
}
-
\cancel{
\mathbb{E}\!\left[
\omega h_{x^*}(x^*)
\,\middle|\, \theta > \bar\theta
\right]
} \\[0.8em]
&=
\mathbb{E}\!\left[
\theta f_x(x^*)
\,\middle|\, \theta < \bar\theta
\right]
-
\mathbb{E}\!\left[
\theta f_x(x^*)
\,\middle|\, \theta > \bar\theta
\right] .
\end{aligned}
\end{equation}

The concavity of $f(x)$ leads to $f_x(x)>0$.
Equation~\ref{eq-compz1} and Equation~\ref{eq-compz2} can then be signed
to be negative so that marginal profit in the good state is always
higher when insurance contracts are triggered on $\theta$.

\[
\frac{\partial \mathbb{E}[\pi|\theta<\bar\theta]}{\partial x^*}-\frac{\partial \mathbb{E}[\pi|\theta>\bar\theta]}{\partial x^*}=\overbrace{\mathbb{E}[\theta f_{x^*}(x^*)|\theta<\bar\theta]-\mathbb{E}[\theta f_{x^*}(x^*)|\theta>\bar\theta]}^{-}
\]
\end{proof}

\subsection{\texorpdfstring{Proof of
Lemma~\ref{lem-mp}}{Proof of Lemma~}}\label{proof-of-lem-mp}

\textbf{Lemma 3.1} \emph{Expected marginal profit is higher in bad
states for risk decreasing inputs when contracts are built on extraction
risk} $\omega$

$\frac{\mathbb{E}[\partial \pi|\omega<\bar \omega]}{\partial x}-\frac{\mathbb{E}[\partial \pi|\omega>\bar \omega]}{\partial x}>0$
\emph{if} $h_{x}(x)<0$.

\emph{Otherwise, risk increasing inputs lead to higher expected marginal
profit in the good states.}

$\frac{\mathbb{E}[\partial \pi|\omega<\bar \omega]}{\partial x}-\frac{\mathbb{E}[\partial \pi|\omega>\bar \omega]}{\partial x}<0$
\emph{if} $h_{x}(x)>0$

\begin{proof}
By the first order conditions, there exist optimal values of any
individual input $x$ that must be chosen before the realization of the
states of the world. Therefore $h(x^*)$, $f(x^*)$, and $c(x^*)$
are equal across states. Where $x^*$ denotes the optimal input choice.

Marginal utility in both states of the world is controlled by risk
effects and the sign of the random variables. Given $\theta$ is
independent of $\omega$, the expected value of
$\mathbb{E}[\theta|\omega\lessgtr\bar\omega]=0$. The difference in
expected marginal profit across insurance states is defined as:

\begin{equation}\label{eq-comppi1}
\begin{aligned}
\frac{\partial \mathbb{E}[\pi \mid \omega < \bar\omega]}{\partial x^*}
-
\frac{\partial \mathbb{E}[\pi \mid \omega > \bar\omega]}{\partial x^*}
&=
\mathbb{E}\!\left[
\omega h_{x^*}(x^*)
\,\middle|\, \omega < \bar\omega
\right]
- 
\mathbb{E}\!\left[
\omega h_{x^*}(x^*)
\,\middle|\, \omega > \bar\omega
\right] \\
&\quad
+\, \cancel{f_{x^*}(x^*) \hat{B}}
- \cancel{c_{x^*}(x^*)} \\
&\quad
+\, \cancel{
\mathbb{E}\!\left[
\theta f_x(x^*)
\,\middle|\, \omega < \bar\omega
\right]
}
-
\cancel{
\mathbb{E}\!\left[
\theta f_x(x^*)
\,\middle|\, \omega > \bar\omega
\right]
} \\[0.8em]
&=
\mathbb{E}\!\left[
\omega h_{x^*}(x^*)
\,\middle|\, \omega < \bar\omega
\right]
-
\mathbb{E}\!\left[
\omega h_{x^*}(x^*)
\,\middle|\, \omega > \bar\omega
\right] .
\end{aligned}
\end{equation}

If an input is risk decreasing then $h_{x}(x)<0$. Then
Equation~\ref{eq-comppi1} is positive and marginal profit in the bad
state is greater than the marginal profit in the good state. Adding more
of a risk reducing input reduces the negative impact in the bad state
relative to the good state.

\[
\frac{\partial \mathbb{E}[\pi|\omega<\bar\omega]}{\partial x^*}-\frac{\partial \mathbb{E}[\pi|\omega>\bar\omega]}{\partial x^*}=\overbrace{\mathbb{E}[\omega h_{x^*}(x^*)|\omega<\bar\omega]-\mathbb{E}[\omega h_{x^*}(x^*)|\omega>\bar\omega]}^{+}
\]

Repeating the same steps for risk increasing inputs $h_{x}(x)>0$ shows
that marginal profit in the bad state is less than marginal profit in
the good state.

\[
\frac{\partial \mathbb{E}[\pi|\omega<\bar\omega]}{\partial x^*}-\frac{\partial \mathbb{E}[\pi|\omega>\bar\omega]}{\partial x^*}=\overbrace{\mathbb{E}[\omega h_{x^*}(x^*)|\omega<\bar\omega]-\mathbb{E}[\omega h_{x^*}(x^*)|\omega>\bar\omega]}^{-}
\]
\end{proof}

\subsection{\texorpdfstring{Proof of
Proposition~\ref{prp-samre}}{Proof of Proposition~}}\label{sec-samre}

\textbf{Proposition 4.1} \emph{In fisheries with two inputs, index
insurance specified with contracts on} $\omega$ \emph{will increase
(decrease) the optimal use of a specific input if the input's risk
effects are increasing (decreasing) when the following sufficient
condition is true:}

$\frac{\partial U}{\partial x_a x_b}>0$ \emph{when both inputs share
the same risk effects, and}
$\frac{\partial U}{\partial x_a \partial x_b}<0$\emph{when inputs
have opposite risk effects.}

\emph{Otherwise, index insurance will have ambiguous effects on optimal
input choice.}

\begin{proof}
We use the same insurance design from Section~\ref{sec-common}. Fishers
now maximize expected utility by selecting two inputs. Contracts are
built on $\omega$.

\begin{equation}\label{eq-max2}{
\begin{aligned}
U\equiv\max_{x_a,x_b}\mathbb{E}[U]=\int^{\infty}_{-\infty}&\left[ \int^{\bar \omega}_{-\infty}j(\omega,\theta)u(\pi(X,\hat{B},\theta,\omega)+(1-J(\bar \omega))\gamma)d\omega \right.\\
&\left.+\int^{\infty}_{\bar{\omega}}j(\omega,\theta) u(\pi(X,\hat{B},
\theta,\omega)-J(\bar \omega)\gamma)d\omega\right] d\theta
\end{aligned}
}\end{equation}

Taking the first order conditions yields:

\begin{equation}\label{eq-foc2}
\begin{aligned}
\frac{\partial U}{\partial x_a}
&= \int_{-\infty}^{\infty}
\Biggl[
\int_{-\infty}^{\bar\omega}
j(\omega,\theta)\,
u_{x_a}\!\Bigl(
\pi(X,\hat{B},\theta,\omega)
+ (1 - J(\bar\omega))\gamma
\Bigr) \\
&\qquad\qquad\quad
\times
\frac{\partial \pi}{\partial x_a}
(X,\hat{B},\theta,\omega)
\, d\omega \\
&\quad
+ \int_{\bar\omega}^{\infty}
j(\omega,\theta)\,
u_{x_a}\!\Bigl(
\pi(X,\hat{B},\theta,\omega)
- J(\bar\omega)\gamma
\Bigr) \\
&\qquad\qquad\quad
\times
\frac{\partial \pi}{\partial x_a}
(X,\hat{B},\theta,\omega)
\, d\omega
\Biggr]
\, d\theta
= 0 , \\[1.2em]
\frac{\partial U}{\partial x_b}
&= \int_{-\infty}^{\infty}
\Biggl[
\int_{-\infty}^{\bar\omega}
j(\omega,\theta)\,
u_{x_b}\!\Bigl(
\pi(X,\hat{B},\theta,\omega)
+ (1 - J(\bar\omega))\gamma
\Bigr) \\
&\qquad\qquad\quad
\times
\frac{\partial \pi}{\partial x_b}
(X,\hat{B},\theta,\omega)
\, d\omega \\
&\quad
+ \int_{\bar\omega}^{\infty}
j(\omega,\theta)\,
u_{x_b}\!\Bigl(
\pi(X,\hat{B},\theta,\omega)
- J(\bar\omega)\gamma
\Bigr) \\
&\qquad\qquad\quad
\times
\frac{\partial \pi}{\partial x_b}
(X,\hat{B},\theta,\omega)
\, d\omega
\Biggr]
\, d\theta
= 0 .
\end{aligned}
\end{equation}

Assuming the first order condition is satisfied, we can use the implicit
function theorem (IFT) to look at the impact of a change in the
exogenous insurance contract. Applying IFT yields a system of equations
that determine the impact of insurance on each optimal input:

\begin{equation}\label{eq-ivtsol}{
\begin{aligned}
&\frac{\partial x_a}{\partial \gamma}=\frac{-1}{Det}\left[\frac{\partial U}{\partial x_b \partial x_b}\frac{\partial U}{\partial x_a \partial \gamma}-\frac{\partial U}{\partial x_a \partial x_b}\frac{\partial U}{\partial x_b \partial \gamma}\right] \\
&\frac{\partial x_b}{\partial \gamma}=\frac{-1}{Det}\left[\frac{-\partial U}{\partial x_b \partial x_a}\frac{\partial U}{\partial x_a \partial \gamma}+\frac{\partial U}{\partial x_a \partial x_a}\frac{\partial U}{\partial x_b \partial \gamma}\right]
\end{aligned}
}\end{equation}

Because the determinant (DET) will always be positive by the
second-order condition, we can focus on the interior of the brackets. If
positive, then insurance will lower use of that specific input and vice
versa. The partial derivatives in Equation~\ref{eq-ivtsol} are complex.
Their complete derivations are included in Section~\ref{sec-partial}.

Lemma~\ref{lem-mp} allows us to sign the partial equations
Equation~\ref{eq-kgam} and Equation~\ref{eq-lgam} for any risk effect on
either input. Concave utility by definition leads to $u''<0$. For
simplicity, we will only focus on
$\frac{\partial U}{\partial x_a\partial \gamma}$, but all applies
equally to $\frac{\partial U}{\partial x_b\partial \gamma}$. Insurance
payouts equalize profits between different states. If insurance
completely covers all loss and $x_a$ is risk increasing, then
$\frac{\partial U}{\partial x_a\partial \gamma}$ is positive.

\begin{equation}\label{eq-kgamsol}{
\begin{aligned}
\frac{U}{\partial x_a \partial \gamma}=\int^{\infty}_{-\infty}&\overbrace{j_{\theta}(\theta)J(\bar\theta)(1-J(\bar\theta))u''(\theta,\cdot)}^{-}\\
&\left[ \int^{\bar\omega}_{-\infty}\underbrace{j_{\omega}(\omega)\frac{\partial \pi}{\partial x_a}d\omega
-\int^{\infty}_{\bar{\theta}}j_{\omega}(\omega)\frac{\partial \pi}{\partial x_a}d\omega}_{-}\right] d\theta\\
>0
\end{aligned}
}\end{equation}

Suppose both inputs are risk increasing so
$\frac{\partial U}{\partial x_a\partial \gamma}$ and
$\frac{\partial U}{\partial x_b\partial \gamma}$ are positive. The
only way for Equation~\ref{eq-ivtsol} to be unambiguously positive is
for $\frac{\partial U}{\partial x_a\partial x_b}$ and
$\frac{\partial U}{\partial x_b\partial x_a}$ to be positive.

\[
\begin{aligned}
&\frac{\partial x_a}{\partial \gamma}=\overbrace{\frac{-1}{Det}}^{-}\left[\overbrace{\overbrace{\frac{\partial U}{\partial x_b \partial x_b}}^{-}\overbrace{\frac{\partial U}{\partial x_a \partial \gamma}}^{+}\overbrace{-\frac{\partial U}{\partial x_a \partial x_b}}^{-}\overbrace{\frac{\partial U}{\partial x_b \partial \gamma}}^{+}}^{-}\right] >0\\
&\frac{\partial x_b}{\partial \gamma}=\overbrace{\frac{-1}{Det}}^{-}\left[\overbrace{\overbrace{\frac{-\partial U}{\partial x_b \partial x_a}}^{-}\overbrace{\frac{\partial U}{\partial x_a \partial \gamma}}^{+}+\overbrace{\frac{\partial U}{\partial x_a \partial x_a}}^{-}\overbrace{\frac{\partial U}{\partial x_b \partial \gamma}}^{+}}^{-}\right]>0
\end{aligned}
\]

Both risk increasing inputs will be raised with index insurance.
Repeating the same steps above with risk decreasing inputs shows both
inputs unambiguously decrease with index insurance.

Now suppose inputs have mixed risk effects. For simplicity, $x_a$ will
be risk increasing and $x_b$ will be risk decreasing. The results will
hold for the opposite case. By Lemma~\ref{lem-mp},
$\frac{\partial U}{\partial x_a\partial \gamma}$ is positive, while
$\frac{\partial U}{\partial x_b\partial \gamma}$ is negative.
Equation~\ref{eq-ivtsol} will be unambiguously positive if
$\frac{\partial U}{\partial x_a\partial x_b}$ and
$\frac{\partial U}{\partial x_b\partial x_a}$ are negative.

\[
\begin{aligned}
&\frac{\partial x_a}{\partial \gamma}=\overbrace{\frac{-1}{Det}}^{-}\left[\overbrace{\overbrace{\frac{\partial U}{\partial x_b\partial x_b}}^{-}\overbrace{\frac{\partial U}{\partial x_a \partial \gamma}}^{+}\overbrace{-\frac{\partial U}{\partial x_a \partial x_b}}^{+}\overbrace{\frac{\partial U}{\partial x_b\partial \gamma}}^{-}}^{-}\right] >0\\
&\frac{\partial x_b}{\partial \gamma}=\overbrace{\frac{-1}{Det}}^{-}\left[\overbrace{\overbrace{\frac{-\partial U}{\partial x_b\partial x_a}}^{+}\overbrace{\frac{\partial U}{\partial x_a \partial \gamma}}^{+}+\overbrace{\frac{\partial U}{\partial x_a \partial x_a}}^{-}\overbrace{\frac{\partial U}{\partial x_b\partial \gamma}}^{-}}^{+}\right]<0
\end{aligned}
\]

The risk increasing input will be raised with index insurance, while the
risk decreasing input will be lowered.

If these conditions do not hold, then it is impossible to determine
which additive element outweighs the other, and the insurance effects on
optimal input use will be ambiguous regardless of the underlying risk
effects of an input.
\end{proof}

\subsection{Partial derivatives}\label{sec-partial}

Partial derivatives used to sign Equation~\ref{eq-ivtsol} are shown
below. For brevity, $\pi(X,\hat B,\omega,\theta)$ is reduced to
$\pi$.

\begin{equation}\label{eq-kk}
\begin{aligned}
\frac{\partial^2 U}{\partial x_a \, \partial x_a}
&= \int_{-\infty}^{\infty}
\Biggl[
\int_{-\infty}^{\bar\omega}
j(\omega,\theta)
\Biggl\{
u''\!\bigl(
\pi + (1 - J(\bar\omega))\gamma
\bigr)
\frac{\partial \pi}{\partial x_a} \\
&\qquad\qquad\quad
+\, u'\!\bigl(
\pi + (1 - J(\bar\omega))\gamma
\bigr)
\frac{\partial^2 \pi}{\partial x_a^2}
\Biggr\}
\, d\omega \\
&\quad
+ \int_{\bar\omega}^{\infty}
j(\omega,\theta)
\Biggl\{
u''\!\bigl(
\pi - J(\bar\omega)\gamma
\bigr)
\frac{\partial \pi}{\partial x_a} \\
&\qquad\qquad\quad
+\, u'\!\bigl(
\pi - J(\bar\omega)\gamma
\bigr)
\frac{\partial^2 \pi}{\partial x_a^2}
\Biggr\}
\, d\omega
\Biggr]
\, d\theta .
\end{aligned}
\end{equation}

\begin{equation}\label{eq-ll}
\begin{aligned}
\frac{\partial^2 U}{\partial x_b \, \partial x_b}
&= \int_{-\infty}^{\infty}
\Biggl[
\int_{-\infty}^{\bar\omega}
j(\omega,\theta)
\Biggl\{
u''\!\bigl(
\pi + (1 - J(\bar\omega))\gamma
\bigr)
\frac{\partial \pi}{\partial x_b} \\
&\qquad\qquad\quad
+\, u'\!\bigl(
\pi + (1 - J(\bar\omega))\gamma
\bigr)
\frac{\partial^2 \pi}{\partial x_b^2}
\Biggr\}
\, d\omega \\
&\quad
+ \int_{\bar\omega}^{\infty}
j(\omega,\theta)
\Biggl\{
u''\!\bigl(
\pi - J(\bar\omega)\gamma
\bigr)
\frac{\partial \pi}{\partial x_b} \\
&\qquad\qquad\quad
+\, u'\!\bigl(
\pi - J(\bar\omega)\gamma
\bigr)
\frac{\partial^2 \pi}{\partial x_b^2}
\Biggr\}
\, d\omega
\Biggr]
\, d\theta .
\end{aligned}
\end{equation}

\begin{equation}\label{eq-crossl}
\begin{aligned}
\frac{\partial^2 U}{\partial x_a \, \partial x_b}
&= \int_{-\infty}^{\infty}
\Biggl[
\int_{-\infty}^{\bar\omega}
j(\omega,\theta)
\Biggl\{
u''\!\bigl(
\pi + (1 - J(\bar\omega))\gamma
\bigr)
\frac{\partial \pi}{\partial x_a}
\frac{\partial \pi}{\partial x_b} \\
&\qquad\qquad\quad
+\, u'\!\bigl(
\pi + (1 - J(\bar\omega))\gamma
\bigr)
\frac{\partial^2 \pi}{\partial x_a \, \partial x_b}
\Biggr\}
\, d\omega \\
&\quad
+ \int_{\bar\omega}^{\infty}
j(\omega,\theta)
\Biggl\{
u''\!\bigl(
\pi - J(\bar\omega)\gamma
\bigr)
\frac{\partial \pi}{\partial x_a}
\frac{\partial \pi}{\partial x_b} \\
&\qquad\qquad\quad
+\, u'\!\bigl(
\pi - J(\bar\omega)\gamma
\bigr)
\frac{\partial^2 \pi}{\partial x_a \, \partial x_b}
\Biggr\}
\, d\omega
\Biggr]
\, d\theta .
\end{aligned}
\end{equation}

\begin{equation}\label{eq-crossk}
\begin{aligned}
\frac{\partial^2 U}{\partial x_b \, \partial x_a}
&= \int_{-\infty}^{\infty}
\Biggl[
\int_{-\infty}^{\bar\omega}
j(\omega,\theta)
\Biggl\{
u''\!\bigl(
\pi + (1 - J(\bar\omega))\gamma
\bigr)
\frac{\partial \pi}{\partial x_a}
\frac{\partial \pi}{\partial x_b} \\
&\qquad\qquad\quad
+\, u'\!\bigl(
\pi + (1 - J(\bar\omega))\gamma
\bigr)
\frac{\partial^2 \pi}{\partial x_b \, \partial x_a}
\Biggr\}
\, d\omega \\
&\quad
+ \int_{\bar\omega}^{\infty}
j(\omega,\theta)
\Biggl\{
u''\!\bigl(
\pi - J(\bar\omega)\gamma
\bigr)
\frac{\partial \pi}{\partial x_a}
\frac{\partial \pi}{\partial x_b} \\
&\qquad\qquad\quad
+\, u'\!\bigl(
\pi - J(\bar\omega)\gamma
\bigr)
\frac{\partial^2 \pi}{\partial x_b \, \partial x_a}
\Biggr\}
\, d\omega
\Biggr]
\, d\theta .
\end{aligned}
\end{equation}
\begin{equation}\label{eq-kgam}{
\begin{aligned}
\frac{\partial U}{\partial x_a \partial \gamma}=\int^\infty_{-\infty}&\left[ \int^{\bar \omega}_{-\infty}j(\omega,\theta)u''(\pi+(1-J(\bar \omega))\gamma)\frac{\partial\pi}{\partial x_a}(1-J(\bar \omega)d\omega \right. \\
& \left.\int^\infty_{\bar \omega}j(\omega,\theta)u''(\pi-J(\bar \omega)\gamma)\frac{\partial\pi}{\partial x_a}(-J(\bar \omega))d\omega\right]d\theta
\end{aligned}
}\end{equation}

\begin{equation}\label{eq-lgam}{
\begin{aligned}
\frac{\partial U}{\partial x_b \partial \gamma}=\int^\infty_{-\infty}&\left[ \int^{\bar \omega}_{-\infty}j(\omega,\theta)u''(\pi+(1-J(\bar \omega))\gamma)\frac{\partial\pi}{\partial x_b}(1-J(\bar \omega)d\omega \right. \\
&\left.\int^\infty_{\bar \omega}j(\omega,\theta)u''(\pi-J(\bar \omega)\gamma)\frac{\partial\pi}{\partial x_b}(-J(\bar \omega))d\omega\right]d\theta
\end{aligned}
}\end{equation}

\subsection{Correlation Proofs}\label{correlation-proofs}

\begin{lemma}[]\protect\hypertarget{lem-corr}{}\label{lem-corr}

When shocks are perfectly correlated, expected marginal profit is always
higher in the good state when an input, $x$, is risk increasing and
ambiguous when $x$ is risk decreasing. This hold regardless of the
chosen index.

$\frac{\mathbb{E}[\partial \pi|\omega<\bar \omega]}{\partial x}-\frac{\mathbb{E}[\partial \pi|\omega>\bar \omega]}{\partial x}<0$
if $h_{x}(X)>0$

And,
$\frac{\mathbb{E}[\partial \pi|\omega<\bar \omega]}{\partial x}-\frac{\mathbb{E}[\partial \pi|\omega>\bar \omega]}{\partial x}\lessgtr 0$
if $h_{x_m}(X)<0$.

\end{lemma}

\begin{proof}
Perfect correlation between two random variables centered at 0 imply
that whenever one variable is negative, so too is the other. Due to
this, we focus only on $\omega$ as the index. The proof follows
identically if replaced by an index on $\theta$.

\begin{equation}\label{eq-corrmp}{
\begin{aligned}
\frac{\partial \mathbb{E}[\pi|\omega<\bar\omega]}{\partial x}-&\frac{\partial \mathbb{E}[\pi|\omega>\bar\omega]}{\partial x}=\\ &\cancel{f_{x}(x)}\hat{B}+\mathbb{E}[\theta f_x(x)|\omega<\bar\omega]+\mathbb{E}[\omega h_{x}(x)|\omega<\bar\omega]-\cancel c(x)\\
&-\cancel{f_{x}(x)}\hat{B}+\mathbb{E}[\theta f_x(x)|\omega>\bar\omega]+\mathbb{E}[\omega h_{x}(x)|\omega>\bar\omega]-\cancel c(x)
\end{aligned}
}\end{equation}

When $h_x(X)>0$, Equation~\ref{eq-corrmp} is always negative. Expected
marginal profit is always higher in the good trigger state when shocks
are perfectly correlated.

When $h_x(X)<0$, Equation~\ref{eq-corrmp} is ambiguous. The sign of
each line depends on the relative effect between $f_x(X)$ and
$h_x(X)$. If the risk effects term dominates then
Equation~\ref{eq-corrmp} will be positive. Without further information
it is impossible to know which effect dominates.
\end{proof}

\begin{proposition}[]\protect\hypertarget{prp-corr}{}\label{prp-corr}

For feasible index insurance contracts specified at either trigger,
$\bar\omega=0$ or $\bar\theta=0$, when $\omega$ and $\theta$ are
perfectly correlated random variables, the change in the optimal input
is ambiguous when $h_x(x)<0$ and increases when $h_x(x)>0$.

\end{proposition}

\begin{proof}
Perfect correlation implies $\theta<0$ when $\omega<0$ and
$\theta>0$ when $\omega>0$ since both distributions have mean zero,
$\mathbb{E}[\theta]\equiv\mathbb{E}[\omega]=0$. The bounds of the
integral can be with respect to either trigger. For simplicity, we will
use $\bar\omega$ as the trigger, but the proof holds with
$\bar\theta$.

\begin{equation}\label{eq-per}
\begin{aligned}
\frac{\partial^2 U}{\partial x \, \partial \gamma}
&= (1 - J(\bar\omega))
\int_{-\infty}^{\bar\omega}
\int_{-\infty}^{\bar\omega}
j(\omega,\theta)\,
u''\!\Bigl(
\pi(x,\hat{B},\theta,\omega)
+ (1 - J(\bar\omega))\gamma
\Bigr) \\
&\qquad \times
\frac{\partial \pi}{\partial x}
(x,\hat{B},\theta,\omega)
\, d\omega \, d\theta \\
&\quad
- J(\bar\omega)
\int_{\bar\omega}^{\infty}
\int_{\bar\omega}^{\infty}
j(\omega,\theta)\,
u''\!\Bigl(
\pi(x,\hat{B},\theta,\omega)
- J(\bar\omega)\gamma
\Bigr) \\
&\qquad \times
\frac{\partial \pi}{\partial x}
(x,\hat{B},\theta,\omega)
\, d\omega \, d\theta .
\end{aligned}
\end{equation}

Suppose insurance fully covers the loss between states, then utility in
the good state and bad state are equal to each other so that we can
factor out like terms in Equation~\ref{eq-per}.

\begin{equation}\label{eq-corrsol}
\begin{aligned}
\frac{\partial^2 U}{\partial x \, \partial \gamma}
&= u''(\cdot)\, J(\bar\omega)\bigl(1 - J(\bar\omega)\bigr) \\
&\quad \times \Biggl[
  \int_{-\infty}^{\bar\omega}
  \int_{-\infty}^{\bar\omega}
  j(\omega,\theta)
  \frac{\partial \pi}{\partial x}
  (x,\hat{B},\theta,\omega)
  \, d\omega \, d\theta \\
&\qquad
  - \int_{\bar\omega}^{\infty}
    \int_{\bar\omega}^{\infty}
    j(\omega,\theta)
    \frac{\partial \pi}{\partial x}
    (x,\hat{B},\theta,\omega)
    \, d\omega \, d\theta
\Biggr].
\end{aligned}
\end{equation}

By Lemma~\ref{lem-corr}, when $h_x(x)<0$ the interior is ambiguous so
Equation~\ref{eq-corrsol} cannot be signed, but is unambiguously
positive when $h_x(x)>0$.
\end{proof}

\bibliographystyle{elsarticle-num-names}
\bibliography{library}

\end{document}